\documentclass{article}

\usepackage{arxiv}

\usepackage[utf8]{inputenc} 
\usepackage[T1]{fontenc}    
\usepackage{url}            
\usepackage{booktabs}       
\usepackage{amsfonts}       
\usepackage{nicefrac}       
\usepackage{microtype}      
\usepackage{lipsum}
\usepackage{graphicx}
\usepackage{natbib}
\usepackage{amsthm}  
\usepackage{amsmath}
\usepackage{amssymb}
\usepackage{array}
\usepackage{boldline}
\usepackage{float}
\newtheorem{theorem}{Theorem}
\newtheorem{example}{Example}
\newtheorem{definition}{Definition}
\newtheorem{proposition}{Proposition}

\graphicspath{ {./images/} }

\title{A Kolmogorov--Arnold Neural Model for \\ Cascading Extremes}

\author{
  Miguel de Carvalho\\
  School of Mathematics\\
  University of Edinburgh\\
  EH9 3FD, Edinburgh, UK \\
  \texttt{Miguel.deCarvalho@ed.ac.uk} \\
   \And
 Clemente Ferrer \\
  Department of Mathematics\\
  Universidad T\'ecnica Federico Santa Maria\\
  Valpara\'iso, Chile \\
  \texttt{Clemente.Ferrer@usm.cl} \\
  \And
 Ronny Vallejos \\
  Department of Mathematics\\
  Universidad T\'ecnica Federico Santa Maria\\
  Valpara\'iso, Chile \\
  \texttt{Ronny.Vallejos@usm.cl} \\
}

\begin{document}
\maketitle
\begin{abstract}
This paper addresses the growing concern of cascading extreme events, such as an extreme earthquake followed by a tsunami, by presenting a novel method for risk assessment focused on these domino effects. The proposed approach develops an extreme value theory framework within a Kolmogorov–Arnold network (KAN) to estimate the probability of one extreme event triggering another, conditionally on a feature vector. An extra layer is added to the KAN architecture to ensure that the parameter of interest lies within the unit interval, and we refer to the resulting neural model as KANE (KAN with Natural Enforcement). The proposed method is backed by exhaustive numerical studies and further illustrated with real-world applications to seismology and climatology.
\end{abstract}


\section{Introduction}
Record-breaking extreme events---such as catastrophic wildfires, unprecedented flooding, intense hurricanes, and unparalleled heatwaves---underscore the urgent need to strengthen our quantitative understanding of these occurrences. Extreme Value Theory (EVT) offers a solid mathematical framework, leveraging regular variation and asymptotic principles to estimate risks of such events by extrapolating beyond the limits of available data, into the tails of a distribution \citep{coles2001, beirlant2004,  dehaan2006, resnick2007, HandbookExtremes2026}.

While it is widely recognized by practitioners that extreme events tend to occur in complex sequential forms \citep{cutter2018, raymond2020}, statistical modelling of this context from an EVT viewpoint is still underdeveloped. Multivariate EVT, though a natural starting point, falls short by: i) disregarding the triggering role of certain events; ii) overlooking the order and sequential nature of extreme event cascades; iii) lacking the ability to model feedback loops between events. 

Inspired by the multivariate EVT framework, this paper introduces a novel concept---the POC (Probability of Cascade) surface---which assesses the probability of domino effects between extreme events conditionally on a covariate or feature vector $\mathbf{x} = (x_1, \dots, x_d)^{T}$. As it will be shown below, the POC surface can be interpreted as the probability of a cascading extremal event, as it quantifies the probability that a trigger event (such as an earthquake exceeding magnitude $u$) results in a follow-up event (like a 
subsequent tsunami) as a function of a covariate. The proposed POC-based approach is fully general in the sense that the focus can be placed beyond the case where follow-up event is binary. In particular, we extend the framework to a multi-class setting, allowing for different types of follow-up extreme events, and to an ordinal context, enabling follow-up extreme events to vary in ordinal severity. The case where the follow-up event is continuous includes as a special case the conditional coefficient of extremal dependence introduced by \cite{lee2024}. 

To learn about the POC surface from the data, we develop a neural model grounded on Kolmogorov's superposition theorem.\footnote{This result is also known as Kolmogorov--Arnold representation theorem. Among other things, it is well-known for refuting a conjecture implicit in Hilbert's 13th problem presented in 1900 at the International Congress of Mathematicians.}
\begin{theorem}[Kolmogorov's superposition theorem]\label{Kthem}
  Let $f:[0, 1]^d \to \mathbb{R}$ be a continuous function. Then, $f$ can be expressed as follows
  \begin{equation}\label{rep}
    f(\mathbf{x}) = \sum_{i=1}^{2 d + 1} \Phi_i^{(2)}\left(\sum_{j=1}^d \Phi_{i, j}^{(1)}(x_j)\right), \qquad \mathbf{x} = (x_1, \dots, x_d)^{T},
  \end{equation}
  for some continuous one-dimensional functions $\Phi_{i, j}^{(1)}$ and $\Phi_{j}^{(2)}$.
\end{theorem}

\noindent As it will be noted below, the beauty of this theorem is equally profound from both Mathematical and AI perspectives.

Mathematically speaking, it shows that any multivariate continuous function can be represented using only sums and univariate functions; the word superposition in this context refers to functions of functions; for example, 
\begin{equation}\label{kol}
  f(x_1, x_2, x_3) = g\big(a(\alpha(x_1), \beta(x_2, x_3)), b(x_1, x_2)\big), 
\end{equation}
is a superposition of univariate and bivariate functions. Hence, in this sense, Theorem~\ref{Kthem} shows that all multivariate continuous functions can be reduced to sums and superpositions of univariate functions. 

From an AI perspective, the theorem reveals a two-layer neural network architecture recently popularized by \cite{liu2024}  following their extension to deeper settings. While multi-layer perceptrons are inspired by the universal approximation theorem \citep[e.g.,][]{berlyand2023, bishop2023}, Kolmogorov–Arnold Networks (KAN) are a novel and fast-evolving addition to the AI toolbox, and are rooted on Theorem~\ref{Kthem}. A particularly impressive aspect of KAN is that they are based on the principle {that} any multivariate continuous function can be expressed exactly using only 
$2d + 1$ outer functions and $d$ inner functions. In addition to the many developments following \citeauthor{liu2024}, it should be noted that other neural approaches based on Theorem~\ref{Kthem} had already appeared in the literature \citep{lin1993, sprecher2002, montanelli2020, fakhoury2022}. 

Motivated by Theorem~\ref{Kthem}, we propose a neural model for the POC surface, constructing a network based on sums and superpositions of univariate functions. To ensure the resulting POC surface remains within the unit interval, an additional layer is incorporated into the architecture. This results in a three-layer design that enforces the unit interval constraint inherent to POC surfaces. The same principle can be applied to enforce any range constraint that the final output of a KAN may need to obey, and we refer to this neural model as KANE (Kolmogorov--Arnold Network with Natural Enforcement). Each univariate function in the Kolmogorov superposition representation of the POC surface is modeled using splines, and we provide theoretical guarantees on the model's flexibility by drawing on approximation results for splines. {See Appendix~A. Our theoretical developments may be of independent interest, as they construct an operator inspired by Theorem~\ref{Kthem}---which we term the Kolmogorov superposition operator---and show that, under mild conditions, guarantees for approximating the inner and outer functions can be translated into guarantees for the target function $f$. Finally, Appendix~B comments on how the proposed framework can be extended to a deep setting by a similar approach as in \cite{liu2024}.}

As a byproduct, this paper contributes to the fast-evolving literature on interfaces between EVT and Machine Learning \citep[e.g.,][and references therein]{bhatia2021, allouche2022ev, karpov2022, ns2024}. {Within this line, key references that relate with our context are those of \cite{jalalzai2018} and \cite{aghbalou2024}, who introduced predictive approaches for binary classification in extreme regions, while \cite{huet2024regressionextremeregions} extended these perspectives to a regression context. Our contribution diverges from these powerful approaches by focusing instead on modeling the probability of cascading extremal events.} 

A further distinction is that our focus differs from that of Hawkes processes \citep{laub2021} in a number of important ways. The main difference is that our approach should be understood as an extension of the notion of tail dependence coefficient---aiming at  quantifying conditional probabilities of extreme events---rather than an attempt to model self-excitation of processes.


The remainder of this paper is structured as follows. In Section~\ref{KAN}, we introduce the proposed methods. Section~\ref{extensions} elaborates on extensions. Section~\ref{simulation} assesses the performance of the proposed methods through a series of numerical experiments conducted on simulated data. Empirical illustrations are provided in Section~\ref{applications}. Technical details are provided in the Appendix, whereas the online supplementary materials includes additional computational experiments and details. 

\section{Neural Modeling of Cascading Extremes}\label{KAN}
\subsection{The POC (Probability of Cascade) Surface}
Let $I = \{I_u: u \in \mathbb{R}\}$ be a Bernoulli process (i.e., a random process {with  Bernoulli} marginal distributions), and let $Y
\sim F_Y$ be a continuous random variable. We start by introducing the following functional, referred to as alpha, which plays a central role in our developments:
\begin{equation}
  \label{alpha}
  \alpha \equiv \alpha_I =\lim_{u\to y^{*}}P(I_u=1 \mid \,Y>u),
\end{equation}
{given that the limit exists.} Here and below, $y^{*}=\sup\{y:F_Y(y)<1\}$ is the right endpoint of $F_Y$. 

Loosely, the parameter $\alpha \in [0, 1]$ in  \eqref{alpha} can be interpreted as the probability of a cascading event, where it quantifies the probability that a \textit{trigger event} (such as an earthquake exceeding magnitude $u$) would result in a \textit{follow-up event} (like a subsequent tsunami $I_u = 1$). The nature of the Bernoulli process $I$ defining the follow-up event opens up a variety of modeling possibilities as illustrated below. Particularly, the Bernoulli process may depend on a continuous variable $Z$.

\begin{example}[Tail dependence coefficient]\normalfont \label{tdc}
  If $I_u= I(Z > u)$, where $Y$ and $Z$ have common distribution, then 
\begin{equation}\label{chi}
   \alpha^{\text{TDC}} \equiv \alpha_I = \lim_{u \to y^*} P(Z > u \mid Y > u).
 \end{equation}
Thus, $\alpha$ in \eqref{alpha} includes the well-known tail dependence coefficient \citep[][Chapter~8]{coles2001} as a special case, specifically when the follow-up event involves $Z$ being extreme, with $Z$ being observable.  \strut \hfill $\square$
\end{example}

\begin{example}[Extremal probabilistic index]\label{epi}\normalfont
If the follow-up event is that $Z$ exceeds $Y$, then $I_u = I(Z > Y)$, and hence
\begin{equation}\label{epi2}
  \alpha^{\text{PI}} \equiv \alpha_I = \lim_{u \to y^*} P(Z > Y \mid Y > u),
\end{equation}which can be regarded as extremal version of the probabilistic index \citep{thas2012}. 
Say, from a reliability analysis viewpoint, the extremal probabilistic index in \eqref{epi2} represents the probability of strength ($Z$) to be larger than stress ($Y$), given that stress is extreme. \strut \hfill $\square$
\end{example}  
Our setup keeps in mind that for some applications the variable $Z$ in the above examples might be latent, but it assumes that the Bernoulli process $I$ is  observable. For instance, in the context of Example~\ref{tdc}, we assume occurrence, or not, of a tsunami ($I$) is observable, even if its intensity level ($Z$) might be unavailable, whereas in the context of Example~\ref{epi}, we assume that the collapse, or not, of a system ($I$) is observable, even if its strength ($Z$) is not observable.

{Some additional remarks on the existence of the limit in \eqref{alpha} are warranted. As noted by \cite{embrechts2016} in the context of the tail dependence coefficient, ``for virtually all copula models used in practice, the limit [...] exists''; this observation extends more generally to $\alpha$. For how to construct an instance in which the limit fails to exist, see \cite{kortschak2009}.}

In practice it is desirable to assess how the $\alpha$ functional may be impacted by a covariate or feature. 
The POC (Probability of Cascade) surface, to be introduced below, naturally extends \eqref{alpha} to this covariate-adjusted framework as follows; without loss of generality, we assume the features are scaled to the unit interval.

\begin{definition}[POC Surface]\label{def1} Let $\mathbf{x} = (x_1, \dots, x_d)^{T} \in [0, 1]^d$. The probability of cascade surface is defined as 
  \begin{equation}    \label{eq:alphat}
    \textsc{poc} = \{(\mathbf{x}, \alpha_{I}(\mathbf{x})): \mathbf{x} \in [0, 1]^d\}, \quad
    \alpha_{I}(\mathbf{x}) = \lim_{u \to y^*} P(I_{u, \mathbf{x}} = 1 \mid Y_{\mathbf{x}} > u),
  \end{equation}
  {given that the limit exists,} 
  where $I = \{I_{u, \mathbf{x}}: (u, \mathbf{x}) \in \mathbb{R} \times [0, 1]^d\}$ is a random field with Bernoulli marginal distributions and $\{Y_{\mathbf{x}}: \mathbf{x} \in [0, 1]^d\}$ is a random field.   
\end{definition}

\noindent In the case of a single covariate, we will refer to the POC surface as the POC curve. Although the POC surface is not invariant to continuous monotone transformations of the trigger event, it remains unchanged if the Bernoulli process threshold is adjusted accordingly.
\begin{proposition}\label{invar}
Let $h: \mathbb{R} \rightarrow \mathbb{R}$ be a strictly increasing continuous function and let
        $$
        \alpha_I^h(\mathbf{x})=\lim _{u \rightarrow h\left(y^*\right)} P\{I_{h^{-1}(u), \mathbf{x}}=1 \mid h(Y_{\mathbf{x}})>u\},
        $$
        {given that the limit exists.}  
        Then, $\alpha_I(\mathbf{x})=\alpha_I^h(\mathbf{x})$, for any $\mathbf{x} \in [0, 1]^d$.
\end{proposition}
\begin{proof}      
First, note that $\{h(Y_{\mathbf{x}})>u\}=\{Y_{\mathbf{x}}>h^{-1}(u)\}$; also, by a change of variables $v = h^{-1}(u)$, it follows that if $u \to h(y^*)$, then $v \to y^*$. Thus, 
\begin{equation*}
\begin{split}
\alpha_I^h(\mathbf{x})&=
\lim _{u \rightarrow h(y^*)} P\{I_{\mathbf{x}, h^{-1}(u)}=1 \mid h(Y_{\mathbf{x}})>u\} \\ &= 
\lim _{u \rightarrow h(y^*)} P\{I_{\mathbf{x}, h^{-1}(u)}=1 \mid Y_{\mathbf{x}}>h^{-1}(u)\} \\ &= 
\lim _{v \rightarrow y^*} P\{I_{\mathbf{x}, v}=1 \mid Y_{\mathbf{x}}>v\} \\ &= \alpha_I(\mathbf{x}).
\end{split}
\end{equation*}
\end{proof}
{Next, we show that whenever $P(I_{\mathbf{x},u} = 1 \mid Y_{\mathbf{x}} > u)$ is smooth in $\mathbf{x}$, then, under mild conditions, this smoothness is inherited by the POC surface.}
{\begin{theorem}\label{conttheo}
Let $\mathbf{x} \in [0,1]^d$, and suppose $\{(I_{\mathbf{x},u}, Y_{\mathbf{x}}) : \mathbf{x} \in [0,1]^d,\; u > 0\}$ is a family of random vectors with realizations on $\{0, 1\} \times \mathbb{R}$. Assume:
\begin{enumerate}
    \item[(a)] The mapping $\mathbf{x} \mapsto P(I_{\mathbf{x},u} = 1 \mid Y_{\mathbf{x}} > u)$ is continuous on $[0,1]^d$, for all $u \in \mathbb{R}$;
    \item[(b)] The limit 
    $\lim_{u \to y^*} P(I_{\mathbf{x},u} = 1 \mid Y_{\mathbf{x}} > u)$ 
    exists and convergence is uniform in $\mathbf{x}$ over $[0,1]^d$.
\end{enumerate}
Then, the mapping $\mathbf{x} \mapsto \alpha(\mathbf{x}) \equiv \lim_{u \to y^*} P(I_{\mathbf{x},u} = 1 \mid Y_{\mathbf{x}} > u)$ is continuous on $[0,1]^d$.
\end{theorem}}

{\begin{proof}
Let $\mathbf{x}_0 \in [0,1]^d$ be arbitrary and fix $\varepsilon > 0$. By uniform convergence in (b), there exists $u_0$ such that for all $u \geq u_0$ and all $\mathbf{x} \in [0,1]^d$,
\begin{equation}\label{unifconv}
\left| P(I_{\mathbf{x},u} = 1 \mid Y_{\mathbf{x}} > u) - \alpha(\mathbf{x}) \right| < \frac{\varepsilon}{3}.
\end{equation}
Fix such a $u \geq u_0$. Then, for any $\mathbf{x}$ if follows that $|\alpha(\mathbf{x}) - \alpha(\mathbf{x}_0)| \leq A + B + C$, where
\begin{equation*}
    \begin{cases}
      A = |\alpha(\mathbf{x}) - P(I_{\mathbf{x},u} = 1 \mid Y_{\mathbf{x}} > u)|, \\
      B = |P(I_{\mathbf{x},u} = 1 \mid Y_{\mathbf{x}} > u) - P(I_{\mathbf{x}_0,u} = 1 \mid Y_{\mathbf{x}_0} > u)|,\\
      C = |P(I_{\mathbf{x}_0,u} = 1 \mid Y_{\mathbf{x}_0} > u) - \alpha(\mathbf{x}_0)|.
    \end{cases}
\end{equation*}
By (b), it follows that $A < \varepsilon / 3$ and $C < \varepsilon / 3$. For $B$, note that (a) implies 
that there exists $\delta > 0$, with $\|\textbf{x} - \textbf{x}_0\| < \delta$, such that $B < \varepsilon/3$. Combining all three terms gives
$|\alpha(\mathbf{x}) - \alpha(\mathbf{x}_0)| < \varepsilon.$ This proves the final result.
\end{proof}
All examples in Section~\ref{simulation} can be easily verified to satisfy the assumptions of Theorem~\ref{conttheo}.
}
\subsection{KANE: KAN with Natural Enforcement}
\subsubsection*{A Kolmogorov--Arnold Approach for Learning from Data}\label{KANE Bernoulli} Kolmogorov's superposition theorem suggests a direct approach to modeling the POC surface using a standard KAN framework, by specifying
\begin{equation}\label{starting}
  \alpha_{I}(\mathbf{x}) = \sum_{i=1}^{2 d + 1} {\Phi}_i^{(2)}\left(\sum_{j=1}^d {\Phi}_{i, j}^{(1)}(x_j)\right).
\end{equation}
While theoretically appealing, the formulation in \eqref{starting} does not ensure that in practice $\alpha_{I}(\mathbf{x})$  lies within its `natural' interval, $[0, 1]$. This arises because the true $\alpha_{I}(\mathbf{x})$ is unknown in practice, requiring the inner and outer functions to be modeled and inferred from the data.

A straightforward solution is to include a final activation function in the architecture to `enforce' this constraint. Let $g:\mathbb{R} \to [0, 1]$ be an activation function, and set
  \begin{equation}\label{kane}
    \alpha_{I}(\mathbf{x}) = {g}\left(\sum_{i=1}^{2 d + 1} {\Phi}_i^{(2)}\left(\sum_{j=1}^d {\Phi}_{i, j}^{(1)}(x_j)\right)\right).
  \end{equation}
  Natural candidates to enforce the unit interval constraint include, for example, the sigmoid activation function, 
  \begin{equation}\label{sig}
    g(x) = \frac{1}{1 + \exp(-x)}.
  \end{equation}
  To distinguish the canonical KAN from our approach of enforcing the unit interval constraint, we refer to \eqref{kane} as KANE (KAN with Natural Enforcement).\footnote{Naturally, for other contexts other range constraints can be similarly enforced using alternative activation functions (e.g., $g(x) > 0$ if the target of interest is positive), but a mapping to the unit interval is all we need for modeling the POC surface.} See Fig.~\ref{fig:architecture} for a schematic representation of the architecture of a three-layer KANE model. In function-matrix notation, KANE can be written as 
  \begin{equation*}
    \alpha_{I}(\mathbf{x}) = g\left(({\boldsymbol\Phi^{(2)}} \circ {\boldsymbol\Phi^{(1)}}) \, (\mathbf{x})\right),  
  \end{equation*}
  where 
  \begin{equation*}
  {\boldsymbol\Phi^{(2)}}=\begin{pmatrix}
    \Phi_{1}^{(2)} \\
    \vdots\\
    \Phi_{2d + 1}^{(2)}\\
  \end{pmatrix}^{T}, \qquad 
    {\boldsymbol\Phi^{(1)}}=\begin{pmatrix}
    \Phi_{1,1}^{(1)} &  \cdots & \Phi_{1,d}^{(1)} \\
    \vdots & \ddots & \vdots\\
    \Phi_{2d + 1,1}^{(1)} & \cdots & \Phi_{2d + 1,d}^{(1)}\\
    \end{pmatrix}.
  \end{equation*}
We use the following terminology for the layers of KANE: 
  $\boldsymbol\Phi^{(1)}$ refers to the \textit{input layer}, $\boldsymbol\Phi^{(2)}$ to the \textit{pre-output layer}, whereas $g$ represents the \textit{$g$-layer} which is responsible for transforming the pre-output into the unit interval. 

\begin{figure}
    \centering
  \textbf{Input layer} \hspace*{1.9cm} \textbf{Pre-output layer} \hspace*{2.0cm} \textbf{$g$-layer} \includegraphics[width=0.9\linewidth]{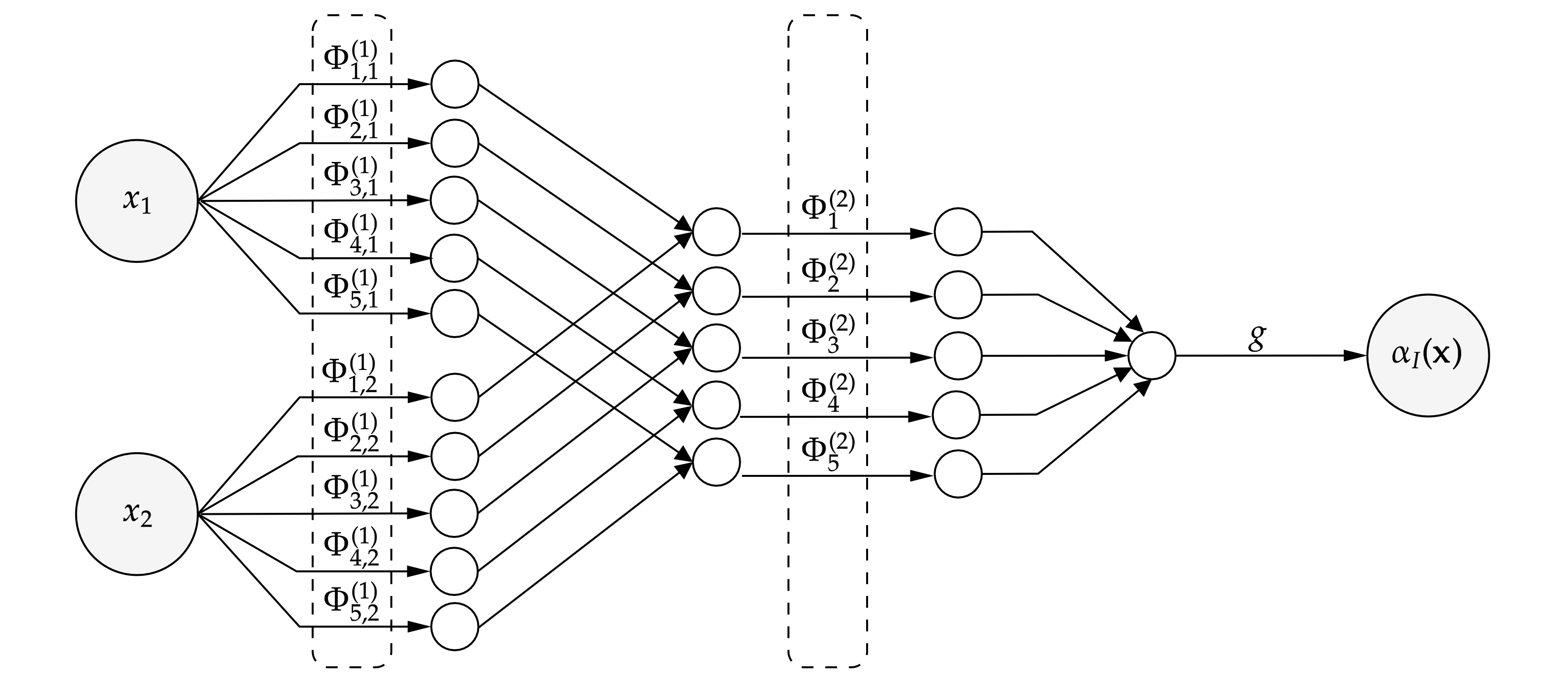}
    \caption{Architecture of three-layer KANE model for the POC surface (binary follow-up event) illustrated on two-feature setting.}
    \label{fig:architecture}
\end{figure}

{A deep version of the proposed framework is outlined in Appendix~B. Motivated however by Theorem~\ref{Kthem}, we mainly focus on the canonical 3-layer KANE model throughout.}

\subsubsection*{Modeling via Splines}
For convenience, in what follows we denote the outer function as $\Phi_i^{(2)}$ or $\Phi_{1, i}^{(2)}$ depending on what is typographically convenient, a choice whose justification will become apparent once the spline formulation and the deep version of the model are introduced. 

Consider $m + 1$ equally-spaced knots, $t_0 < \cdots < t_m$. We model the inner and outer functions as a linear combination of B-spline basis functions, that is, for $l = 1, 2$, 
\begin{equation*}\label{hi}
  \Phi_{i, j}^{(l)}(x) = \sum_{k = 1}^K \beta_{i, j, k}^{(l)} B_k^{\,p}(x), \qquad
\end{equation*}
where  $B_k^{\,p}(x)$ is the B-spline basis function of degree $p$ evaluated at $x$ and $K=p+m$. For $l = 2$ it follows that $i \in \{1, \dots, 2d + 1\}$ and $j \in \{1, \dots, d\}$, whereas for $l = 1$ it holds that $i \in \{1\}$ and $j \in \{1, \dots, 2d + 1\}$. To streamline the notation, we assume the same number of basis functions $K$ per univariate function and the same set of knots, a choice made for convenience that can be adjusted at the cost of more complex notation. The parameter of interest is given by the $(2d + 1) \times (d + 1) \times K$ partitioned tensor $$\boldsymbol{\beta} = \big(\boldsymbol{\beta}^{(1)}_k \mid \boldsymbol{\beta}^{(2)}_k\big),$$ where for $k = 1, \dots, K$, 
\begin{equation*}
  \boldsymbol\beta^{(1)}_k = 
\begin{pmatrix}
\beta_{1,1, k}^{(1)} &  \cdots & \beta_{1,d, k}^{(1)} \\
\vdots & \ddots & \vdots \\
\beta_{2d+1,1, k}^{(1)} & \cdots & \beta_{2d+1,d, k}^{(1)}
\end{pmatrix}
, \qquad
\boldsymbol\beta^{(2)}_k =
\begin{pmatrix}
\beta^{(2)}_{1, 1, k} \\
\vdots \\
\beta^{(2)}_{2d+1, 1, k}
\end{pmatrix}.
\end{equation*}
To learn about the POC surface from data---that is, to infer the tensor $\boldsymbol{\beta}$---we base our inferences on
\begin{equation*}
  D_n = \{(\delta_{i}, \mathbf{x}_i): y_i > u\}_{i = 1}^n,   
\end{equation*}
for a sufficiently large threshold $u$. Here, $\delta_{i}, y_i, \mathbf{x}_i$ are respectively realizations from $I_{u,\mathbf{x}}$ and $Y_{\mathbf{x}}$, given $\mathbf{X} = \mathbf{x}$. The tensor $\boldsymbol{\beta}$ is estimated by minimizing the binary cross-entropy loss, 
\begin{equation*}    L(\boldsymbol{\beta})=-\frac{1}{n}\sum_{i=1}^{n}\{\delta_i\log(\alpha_I(\mathbf{x}_i))+(1-\delta_i)\log(1-\alpha_I(\mathbf{x}_i))\}.
\end{equation*}
The minimizer of $L$, denoted as $\hat{\boldsymbol{\beta}}$, is used to obtain the KANE estimate of the POC surface, \( \hat{\alpha}_I(\mathbf{x}_i) \), as a plug-in estimator in (\ref{kane}). While the asymptotic normality of the resulting estimator remains an open question for future research, the supplementary materials provide numerical evidence that confidence intervals obtained via a resampling cases bootstrap seem to perform reasonably well, with coverage close to the nominal level.


\subsection{Model Checking and Diagnostics}\label{diagnostics}

For evaluating the quality of the fit from the KANE model, we resort to randomized quantile residuals introduced by \cite{dunn1996}. Before we present the version of these residuals for our model, we recall that their basic principle is as follows: For any nonparametric or parametric model, based on a continuous distribution function, $F(x) = P(X \leq x)$, it follows that $R \equiv  \Phi^{-1}(F(X)) \sim N(0, 1)$, when $X \sim F$, where $\Phi^{-1}$ is the standard normal quantile function. Hence, provided that $F$ is a sensible model for the data $X$, then the residuals $R$ {are normally} distributed. As noted by \citeauthor{dunn1996}, the same principle can be easily extended to a discrete model $F$ via jittering. Since for each fixed $u$, our POC surface formulation can be understood as binary response regression model based on $D_n$, it follows that its Dunn--Smyth residuals are defined as 
\begin{equation*}
  R_i =
  \begin{cases}
    \Phi^{-1}(W_i), & \delta_i = 1, \\ 
    \Phi^{-1}(V_i), & \delta_i = 0,
  \end{cases} \quad \text{with } \quad   V_i \sim \text{Unif}(0, 1 - \hat\alpha(\mathbf{x}_i)), \qquad
  W_i \sim \text{Unif}(1 - \hat\alpha(\mathbf{x}_i), 1). 
\end{equation*}
Informally, the $R_i$ are a jittered version of the `binary response' (i.e., follow-up event $\delta_i$), with the level of jittering being controlled by the `probability of success' (i.e., probability of cascade $\hat\alpha(\mathbf{x}_i)$). 
We recommend reporting at least 10 trajectories of the randomized quantile residuals, noting that \cite{dunn1996} suggested that ``four realizations of the quantile residual'' would already be sufficient. These realizations will be presented in a QQ-boxplot along with the confidence bands stemming from the normal reference model (that keeps in mind that the $R_i$ are normally distributed provided that the fit is sensible); the QQ-boxplot was recently introduced by \cite{rodu2021}, and consists of a graphical tool that merges into a single chart a boxplot and a QQ-plot.



\section{Consequences and Extensions}\label{extensions}
\subsection{Multi-Trigger Systems}
In practice, multiple competing incidents can result in the follow-up event. To address this, we define a multi-trigger system where we consider $\mathcal{Y}_1, \dots, \mathcal{Y}_{\mathcal{K} - 1}$ as a sequence of identically distributed random fields with 
$$\mathcal{Y}_k = \{Y_{k, \mathbf{x}}: \mathbf{x} \in [0, 1]^d\}, \qquad k = 1, \dots, \mathcal{K}.$$
Our framework extends to $K$ trigger events by considering
\begin{equation}
  \label{alphamulti}
  \alpha(\mathbf{x}) =\lim_{u\to y^{*}}P(I_{u, \mathbf{x}}=1 \mid \,Y_{1, \mathbf{x}}>u \vee \cdots \vee Y_{\mathcal{K}, \mathbf{x}} > u).
\end{equation}
This formula simplifies to the single-trigger case in \eqref{eq:alphat} by defining $Y_{\mathbf{x}} = \min\{Y_{1, \mathbf{x}}, \dots, Y_{\mathcal{K}, \mathbf{x}}\}$, and hence the theory and methods from Section~\ref{KAN} readily apply to this context as well. The identically distributed assumption for $\mathcal{Y}_1, \dots, \mathcal{Y}_{\mathcal{K}}$ is for convenience, as in practice random processes can be normalized to a common scale (say, unit Fr\'echet margins) using the probability integral transform.  

\subsection{Categorical, Ordinal, and Continuous Follow-up Events}\label{sec:ext2}
In real-world applications, follow-up extreme events can take various forms, including multicategorical, ordinal, and continuous outcomes. As examples of categorical and ordinal events: For instance, $j = 1$ may represent `no tornado,' $j = 2$ a `supercell tornado', and $j = 3$ a `non-supercell tornado'. Similarly, $j = 1$ to $j = 5$ may represent Category 1 to Category 5 hurricanes.

In the multicategorical context, we model the follow-up event as a multinoulli process, that is, $$\mathbf{I} = \{\mathbf{I}_{u, \mathbf{x}}: (u, \mathbf{x}) \in \mathbb{R} \times [0, 1]^d\},
$$
is a random field with multinoulli marginal distribution. Here, $\mathbf{I}_{u, \mathbf{x}} = (I_{u, \mathbf{x}}^{(1)}, \dots, I_{u, \mathbf{x}}^{(J)})^{T} \in \{0, 1\}^J$ and $J$ is the number of categories; see \citet[][Section~2.3.2]{murphy2012} for basics on the  multinoulli distribution. The proposed cascading probability extends to this context as follows.

\begin{definition}[Multicategorical POC Surfaces] Let $\mathbf{x} = (x_1, \dots, x_d)^{T} \in [0, 1]^d$ and $J \in \mathbb{N}$. The $j$th-category probability of cascade surface is defined as 
  \begin{equation}    \label{alpha2}
    \textsc{poc}_j = \{(\mathbf{x}, \alpha_{\mathbf{I}}(\mathbf{x})^{(j)}): \mathbf{x} \in [0, 1]^d\}, \quad
    \alpha_{\mathbf{I}}^{(j)}(\mathbf{x}) = \lim_{u \to y^*} P(\mathbf{I}_{u, \mathbf{x}}^{(j)} = 1 \mid Y_{\mathbf{x}} > u),
  \end{equation}
  for $j \in \{1, \dots, J\}$, {given that the limit exists}, where $\mathbf{I} = \{\mathbf{I}_{u, \mathbf{x}}: (u, \mathbf{x}) \in \mathbb{R} \times [0, 1]^d\}$ is a multinoulli process and $\{Y_{\mathbf{x}}: \mathbf{x} \in [0, 1]^d\}$ is a random field.   
\end{definition}

In terms of fitting the model, 
two adjustments are needed for the architecture discussed in Section~\ref{KAN}. Firstly, we need an activation function for the $g$-layer that maps the Euclidean space to the unit simplex; that is, $\textbf{g} = (g^{(1)}, \dots, g^{(J)})$, where $g^{(j)}: \mathbb{R}^d \to [0, 1]$, such that $\sum_{j = 1}^J g^{(j)}(\mathbf{x}) = 1$. A natural candidate for this task is the softmax activation function, in which case the $g$-layer becomes, 
\begin{equation}\label{soft}
  g^{(j)}(\mathbf{x}) = \frac{\exp(x_j)}{\sum_{i = 1}^J \exp({x_i})}.
\end{equation}
Secondly, the loss function must account for the different categories, for which the cross-entropy loss is employed
\begin{equation*}    L(\boldsymbol{\beta})=-\frac{1}{n}\sum_{i=1}^{n}\sum_{j=1}^{J}\{\delta_{i,j}\log(\alpha_{\mathbf{I}}^{(j)}(\mathbf{x}_i))+(1-\delta_{i,j})\log(1-\alpha_{\mathbf{I}}^{(j)}(\mathbf{x}_i))\},
  \end{equation*}
  where $(\delta_{i, 1}, \dots, \delta_{i, J})^{T}$ is a realization of $\mathbf{I}_{u, \mathbf{x}_i}$.


  Relatedly, the follow-up extremal event may have a natural ordered category. The formal definition of the POC surface for this case is tantamount to the one provided in \eqref{alpha2}. When fitting the model, ordinal events must be handled carefully, avoiding the mistake of treating categories as nominal quantities. For example, a na\"ive approach would be to apply the aforementioned multicategorical framework with a softmax activation function, but this would ignore the ordinal structure in the data. Instead, a simple alternative to address this issue is by devising a version of the well-known approach by \cite{frank2001} to our setting; see also \cite{alcacer2024}. 
  
The Frank–Hall approach begins by transforming a $J$-class ordinal problem into $J-1$ binary classification problems. Specifically, an ordinal attribute $A$ with ordered values $C_1, \dots, C_J$ is converted into $J-1$ binary attributes of the form $A > C_j$, for $j = 1, \dots, J-1$; for example, on the above-mentioned $J = 5$ hurricane problem, this would imply converting the original ordinal problem into four binary attributes based on 
$$ A > \text{Category}~1, \quad A > \text{Category}~2, \quad  A > \text{Category}~3, \quad  A > \text{Category}~4.$$  
Hence, the probabilities for each category for the Frank--Hall approach are given by $P(C_1) = 1 - p_1$, 
\begin{equation*}
    P(C_j) = p_{j-1} - p_j, \quad j=2,\ldots,J-1, 
\end{equation*}
and $P(C_J) = p_{J-1}$, where $p_j = P(A > C_j)$, for all $j$.

Based on this, the ordinal multicategorical POC surface can be derived. Let $\mathbf{A}_{u, \mathbf{x}}$ be an ordinal attribute, and define $\pi_j(\mathbf{x}) = \lim_{u \to y^*} P(\mathbf{A}_{u, \mathbf{x}} > C_1 \mid Y_{\mathbf{x}} > u)$. The Frank--Hall POC surface is then given by  $\alpha_{\mathbf{I}}^{(1)}(\mathbf{x}) = 1 - \pi_j(\mathbf{x})$, 
\begin{equation*}
\alpha_{\mathbf{I}}^{(j)}(\mathbf{x}) = \pi_{j - 1}(\mathbf{x}) - \pi_j(\mathbf{x}), \quad j=2,\ldots,J-1,
\end{equation*}
and $\alpha_{\mathbf{I}}^{(J)}(\mathbf{x}) = \pi_{J - 1}(\mathbf{x})$.

Finally, for a continuous-type follow-up event, if a continuous variable $Z_{\mathbf{x}}$ is observable---similar to Example~\ref{tdc}---we can set $I_{u, x} = I(Z_{\mathbf{x}} > u)$ in which case we obtain the conditional coefficient of extremal dependence by \cite{lee2024}.

\section{Numerical Experiments on Artificial Data}\label{simulation}
\subsection{Artificial Data Generating Processes and Preliminary Experiments}\label{exp}

This section presents the simulation scenarios used to evaluate the proposed methods, along with a series of single-sample experiments. We assess performance in three settings by simulating $n=10\,000$ observations per scenario, with artificial data generated from the Bernoulli process, follow-up events, and features as described below. In each case, we threshold the response at its 95\% quantile ($u$), hence retaining only those values where $Y_x > u$.

\vspace{0.2cm}
\noindent \textbf{Scenarios A1 and A2}: For these scenarios, the follow-up events are respectively simulated from Bernoulli processes with 
\begin{equation*}
    I_{u, x} \mid \, Y_{x} > u \sim \text{Bern}\{m_{\text{A}k}(x; u)\},\quad k=1,2, 
\end{equation*}
for $x \in [0, 1]$. For Scenario~A1, $m_{\text{A1}}(x; u)=\Phi(x, \exp(-u), 1)$ is the distribution function of the normal distribution with mean $\exp(-u)$ and variance 1, whereas for Scenario~A2, $m_{\text{A2}}(x; u)=0.2\sin\{3\pi(x-1)^2\} + 0.4 + 1/u^2$. 

\noindent In both cases, we simulate a single feature $X$ from the standard uniform distribution, and the trigger event $Y_x \mid X = x$ is generated from a unit Fréchet distribution. The true POC curves corresponding to $m_{\text{A1}}(x; u)$ and $m_{\text{A2}}(x; u)$ are respectively given by 
\begin{equation*}
  \begin{cases}
  \alpha^{\text{A1}}_I(x) = \lim_{u \to \infty} m_{\text{A1}}(x; u) = \Phi(x, 0, 1), \\
    \alpha^{\text{A2}}_I(x) = \lim_{u \to \infty} m_{\text{A2}}(x; u) = 0.2\sin\{3\pi(x-1)^2\}+0.4.
  \end{cases}
\end{equation*}

\vspace{0.2cm}
\noindent \textbf{Scenarios B1 and B2}: For these scenarios, the follow-up events are respectively simulated from a Bernoulli processes with 
\begin{equation*}
    I_{u, \mathbf{x}} \mid \, Y_{\mathbf{x}} > u\sim \text{Bern}\{m_{\text{Bk}}(\mathbf{x}; u)\},\quad k=1,2,
\end{equation*}
for $\mathbf{x} = (x_1, x_2) \in [0, 1]^2$. For Scenario~B1, $m_{\text{B1}}(\mathbf{x}; u) = \Phi(\mathbf{x}; \exp(-u)\mathbf{1}_2, \boldsymbol{I}_2)$ is the bivariate distribution function of the bivariate normal distribution, whereas for Scenario~B2, $m_{\text{B2}}(\mathbf{x}; u) = 0.4 \exp(-x_1)\cos\{2\pi x_2\} + 0.5 + 1 / u^2$.
In both cases, we simulate two independent features, $X_1$ and $X_2$, which are drawn from the standard uniform distribution; the trigger event  $Y_{\mathbf{x}} \mid (X_1, X_2) = (x_1, x_2)$ is again simulated from a unit Fréchet distribution. The true POC surfaces corresponding to $m_{\text{B1}}(x; u)$ and $m_{\text{B2}}(x; u)$ are respectively given by 
\begin{equation*}
  \begin{cases}
    \alpha^{\text{B1}}_I(\mathbf{x}) = \lim_{u \to \infty} m_{\text{B1}}(\mathbf{x}; u) = \Phi(\mathbf{x}, 0, \boldsymbol{I}_2), \\ 
    \alpha^{\text{B2}}_I(\mathbf{x}) = \lim_{u \to \infty} m_{\text{B2}}(\mathbf{x}; u) = 0.4\exp(-x_1)\cos\{2\pi x_2\} + 0.5.
  \end{cases}
\end{equation*}

\textbf{Scenario C}: The final scenario examines a multi-category framework, as described in Section~\ref{sec:ext2}. The follow-up events are simulated from a multinoulli process with 
\begin{equation*}
  \mathbf{I}_{u,\mathbf{x}} \mid Y_{\mathbf{x}}>u \sim \text{Multi}_3\left(\frac{m_{\text{C1}}(\mathbf{x}; u)}{M_u},\frac{m_{\text{C2}}(\mathbf{x}; u)}{M_u}, \frac{m_{{\text{C3}}}(\mathbf{x}; u)}{M_u}\right),   
\end{equation*}
for $\mathbf{x} = (x_1, x_2)^{\top} \in [0, 1]^2$. Here, $M_u = m_{\text{C1}}(\mathbf{x}; u) + m_{\text{C2}}(\mathbf{x}; u) + m_{\text{C3}}(\mathbf{x}; u),$
where $m_{\text{C1}}(\mathbf{x}; u) = \Phi(\mathbf{x}; \exp(-u)\mathbf{1}_2, \boldsymbol{I}_2)$, $m_{\text{C2}}(\mathbf{x}; u) = 0.4\exp(-x_1)\cos(2\pi x_2) + 0.5 + 1 / u^2$, and $m_{\text{C3}}(\mathbf{x}; u) = \linebreak 0.8 x_2\sin^2(\pi x_1) + 1 / u^2$. The trigger event and features follow the same settings as in Scenarios B1 and B2 (i.e., the trigger is unit Fréchet-distributed and features are independently drawn from the standard uniform distribution). The true $j$th-category POC surfaces in this case are 
\begin{equation*}
  \begin{cases}
  \begin{aligned}
    \alpha^{(1)}_{\mathbf{I}}(\mathbf{x}) &= \lim_{u \to \infty} \frac{m_{\text{C1}}(\mathbf{x}; u)}{M_u} = \frac{\Phi(\mathbf{x}; \boldsymbol{0}, \boldsymbol{I}_2)}{M}, \\ 
    \alpha^{(2)}_{\mathbf{I}}(\mathbf{x}) &= \lim_{u \to \infty} \frac{m_{\text{C2}}(\mathbf{x}; u)}{M_u} = \frac{0.4\exp(-x_1)\cos\{2\pi x_2\} + 0.5}{M}, \\ 
    \alpha^{(3)}_{\mathbf{I}}(\mathbf{x}) &= \lim_{u \to \infty} \frac{m_{\text{C3}}(\mathbf{x}; u)}{M_u} = \frac{0.8 x_2\sin^2(\pi x_1)}{M}, \\ 
  \end{aligned}
  \end{cases}
\end{equation*}
where
\begin{equation*}
  M = \Phi(\mathbf{x}; \boldsymbol{0}, \boldsymbol{I}_2) + 0.4\exp(-x_1)\cos\{2\pi x_2\} + 0.8 x_2\sin^2(\pi x_1) + 0.5.
\end{equation*}

The results from a single-sample experiment are shown in Fig.~\ref{fig:ABC1}. As can be seen from the latter figures, the proposed methods satisfactorily recover the true POC in all scenarios; these findings are provisional, since they are based on a single-sample experiment; a Monte Carlo evaluation follows in Section~\ref{mcstudy}.

Some final comments on computing and implementations are in order. The neural model was fitted using the LBFGS optimizer \citep[][Chapter~5]{pytlak2008}. For binary classification scenarios, the sigmoid activation function \eqref{sig} was specified for the $g$-layer to ensure that the output values remain within the unit interval, representing valid probabilities. In cases of multi-category classification, the softmax activation function in \eqref{soft} was used. The optimization process was carried out for 100 steps for each scenario. {A key consideration in model fitting was the selection of $p$ and $m$; for the first decision, we relied on the properties of cubic splines and chose $p = 3$, as cubic splines are continuous and have continuous first and second derivatives. The number of knots was set to $m+1=3$, with two knots placed at the boundaries and one interior knot to provide flexibility; this represents the minimal configuration that yielded satisfactory performance in our experiments. 

\subsection{Monte Carlo Evidence}\label{mcstudy}

This section reports the main findings of a Monte Carlo simulation study, where we repeat $500$ times the single-sample experiment from the previous section. We consider three sample sizes $n = 5000, 10\,000,$ and $15\,000$. Figure~\ref{fig:ABCMC} shows the Monte Carlo means of the simulation study, indicating strong performance across different scenarios, as anticipated by Fig.~\ref{fig:ABC1}. 

\begin{table}[H]
    \renewcommand{\arraystretch}{1.5}
    \centering
    \caption{MISE values across scenarios and sample sizes.}
    \label{misetab}
    \begin{tabular}{V{2}c|cccV{2}}
        \hlineB{2}
        \multicolumn{1}{V{2}c|}{\textbf{Scenario}} & \multicolumn{3}{cV{2}}{\textbf{Sample Size}} \\ 
        & $5\,000$ & $10\,000$ & $15\,000$ \\ 
        \hlineB{2}
        \textbf{A1} & $9.112\times10^{-6}$ & $3.706\times10^{-6}$ & $7.759\times10^{-7}$ \\ 
        \textbf{A2} & $1.501\times10^{-4}$ & $1.541\times10^{-4}$ & $1.478\times10^{-4}$ \\ 
        \textbf{B1} & $1.021\times10^{-4}$ & $8.014\times10^{-5}$ & $7.722\times10^{-5}$ \\ 
        \textbf{B2} & $2.669\times10^{-3}$ & $2.626\times10^{-3}$ & $2.625\times10^{-3}$ \\ 
        \textbf{C ($j=1$)} & $9.465\times10^{-4}$ & $9.055\times10^{-4}$ & $8.958\times10^{-4}$ \\ 
        \textbf{C ($j=2$)} & $1.166\times10^{-3}$ & $1.130\times10^{-3}$ & $1.106\times10^{-3}$ \\ 
        \textbf{C ($j=3$)} & $1.725\times10^{-4}$ & $1.185\times10^{-4}$ & $9.428\times10^{-5}$ \\ 
        \hlineB{2}
    \end{tabular}
\end{table}

Finally, in Table~\ref{misetab} we report the Monte Carlo MISE (Mean Integrated Squared Error), where
\begin{equation*}
\text{MISE} = \int_{[0, 1]^d} \{\hat \alpha_I(\mathbf{x}) - \alpha_I(\mathbf{x})\}^2 \, \text{d} \mathbf{x}.
\end{equation*} 
Here, $\hat \alpha_I(\mathbf{x})$ is the KANE estimate of POC and $[0, 1]^d$ is the space of the features (i.e., $d = 1$ for Scenarios A1--A2, and $d = 2$ for Scenarios B1--B2 and C). As can be seen from Table~\ref{misetab}, the proposed approach has appealing frequentist properties from a numerical viewpoint, in the sense that as the sample size increases, MISE decreases. 

In the supplementary materials, we document a battery of additional Monte Carlo experiments, including analyses of dependent features, an evaluation of the deep model, {higher-dimensional features}, and a comparison of the Frank--Hall ordinal model with the softmax approach. Among other findings, these additional results further validate the approach’s performance, show that performance is better in data-rich regions of the covariate space, and that the canonical three-layer model may perform as well as, if not better than, its deep counterpart, {which also requires careful selection of the number of layers to avoid overfitting}.

Moreover, the canonical three-layer KANE is computationally efficient and stable, consistently converging in under one second per scenario on a standard laptop (Apple M3) for $p=2$. However, as noted in the supplementary materials, increasing the feature dimension to $p>5$ leads to a substantial rise in computational time. Finally, we note that estimators for the POC surface can also be constructed using Generalized Additive Models (GAM) and Multi-Layer Perceptrons (MLP). Numerical evidence in the supplementary materials suggests that KANE and MLP yield comparable POC surface estimates, with both outperforming GAM.

\begin{minipage}[t]{0.45\textwidth}
    \begin{figure}[H]
    \centering
    \includegraphics[width=0.95\linewidth]{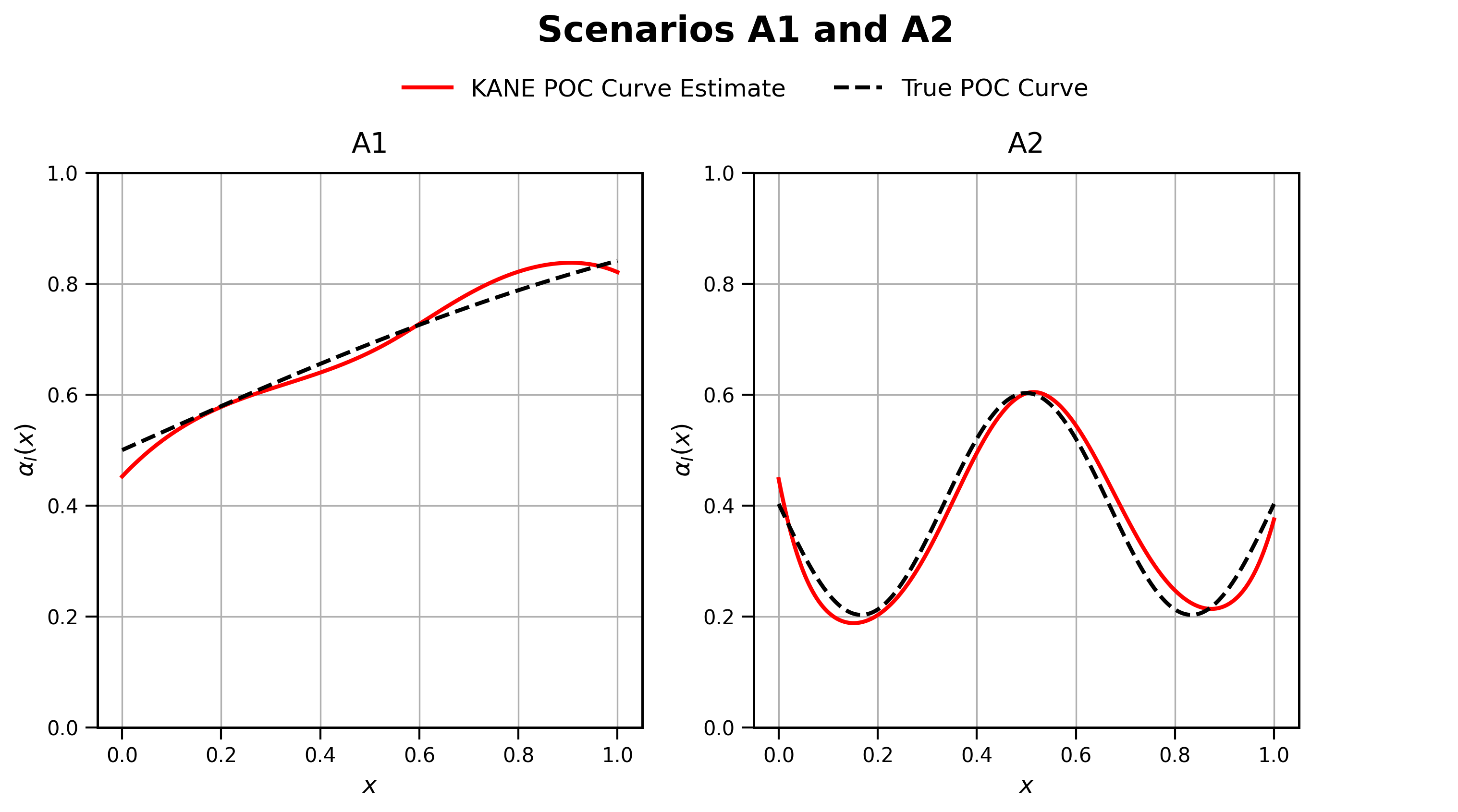}\\
    \includegraphics[width=0.95\linewidth]{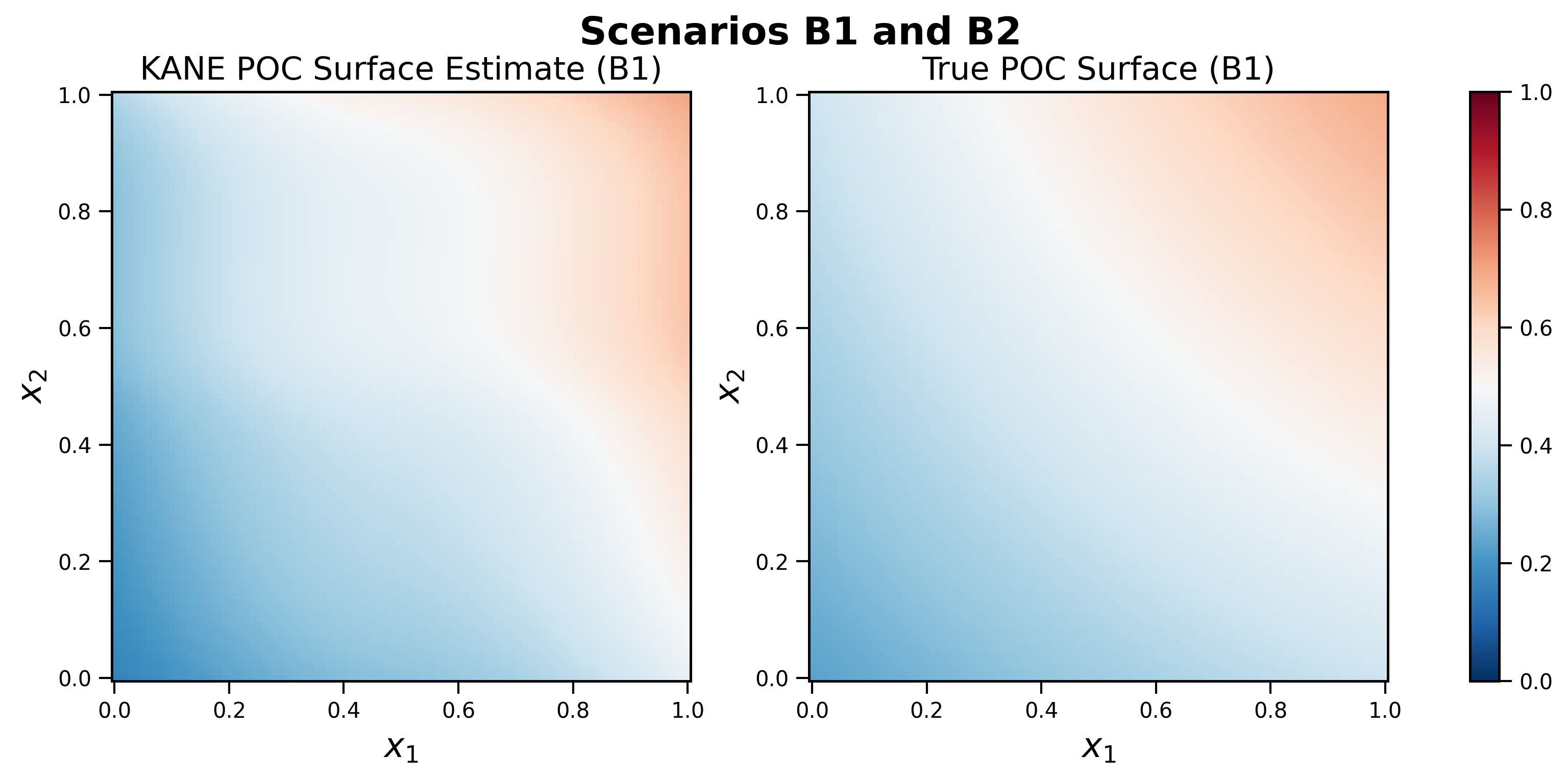}\\
    \includegraphics[width=0.95\linewidth]{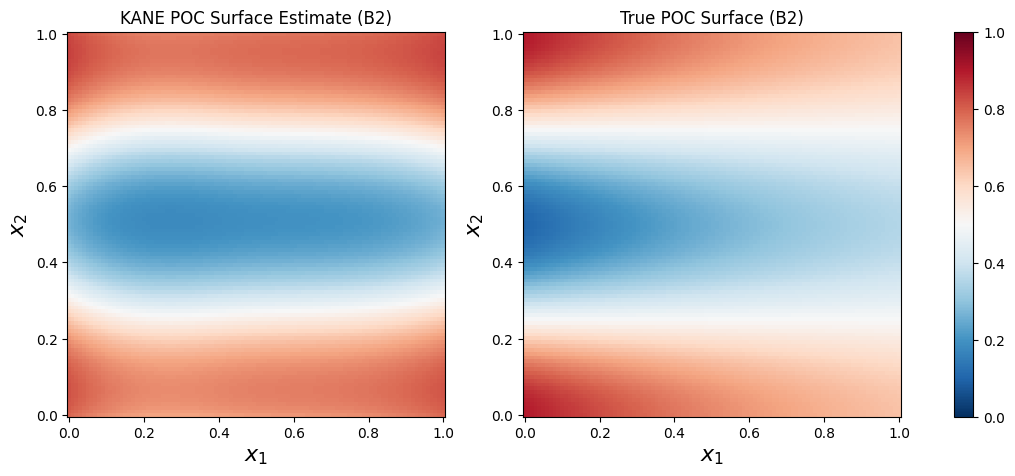}\\
    \includegraphics[width=0.95\linewidth]{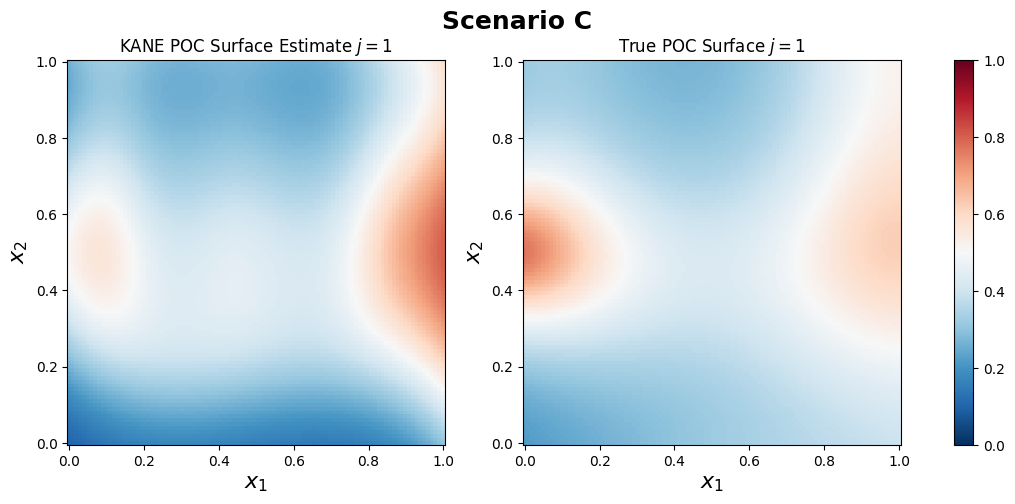}\\
    \includegraphics[width=0.95\linewidth]{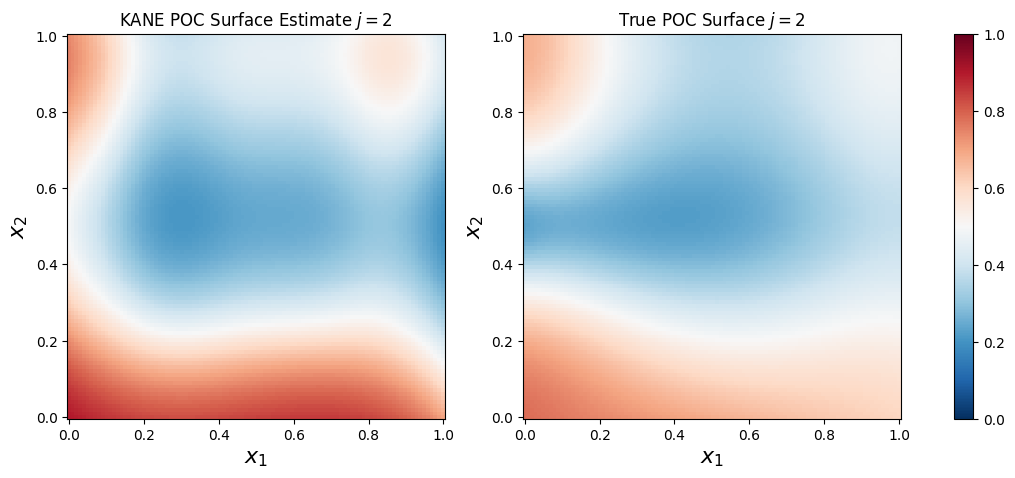}\\
    \includegraphics[width=0.95\linewidth]{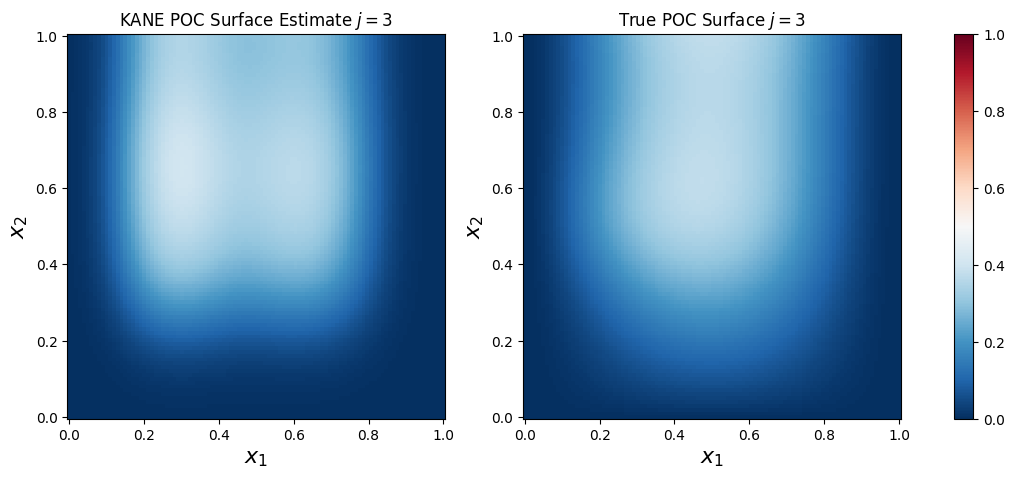}    
    \caption{Single-sample experiments for Scenarios~A, B, and C.}
    \label{fig:ABC1}
\end{figure}
\end{minipage}\hspace{1cm}  
\begin{minipage}[t]{0.45\textwidth}
    \begin{figure}[H]
    \centering
    \includegraphics[width=0.95\linewidth]{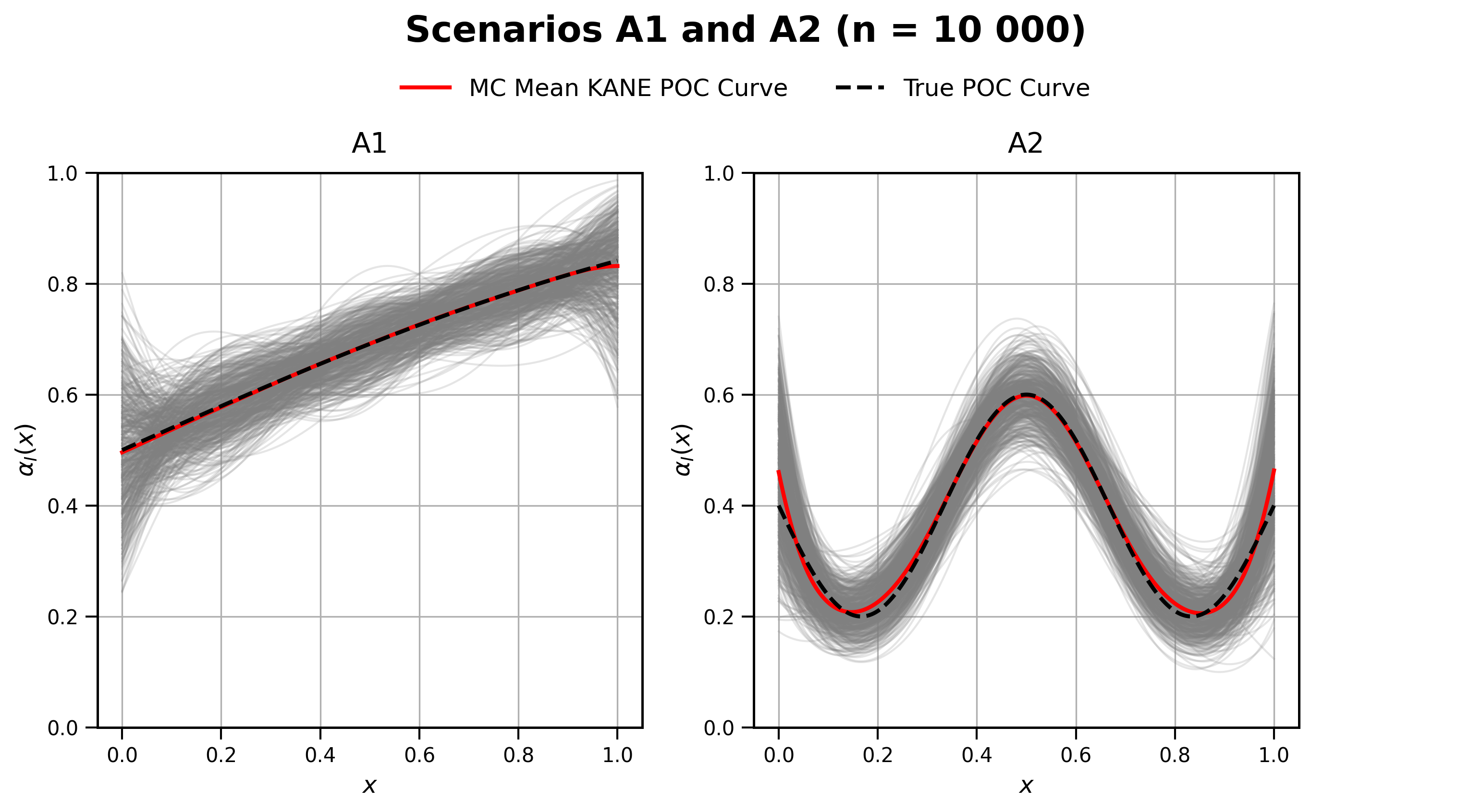}\\
    \includegraphics[width=0.95\linewidth]{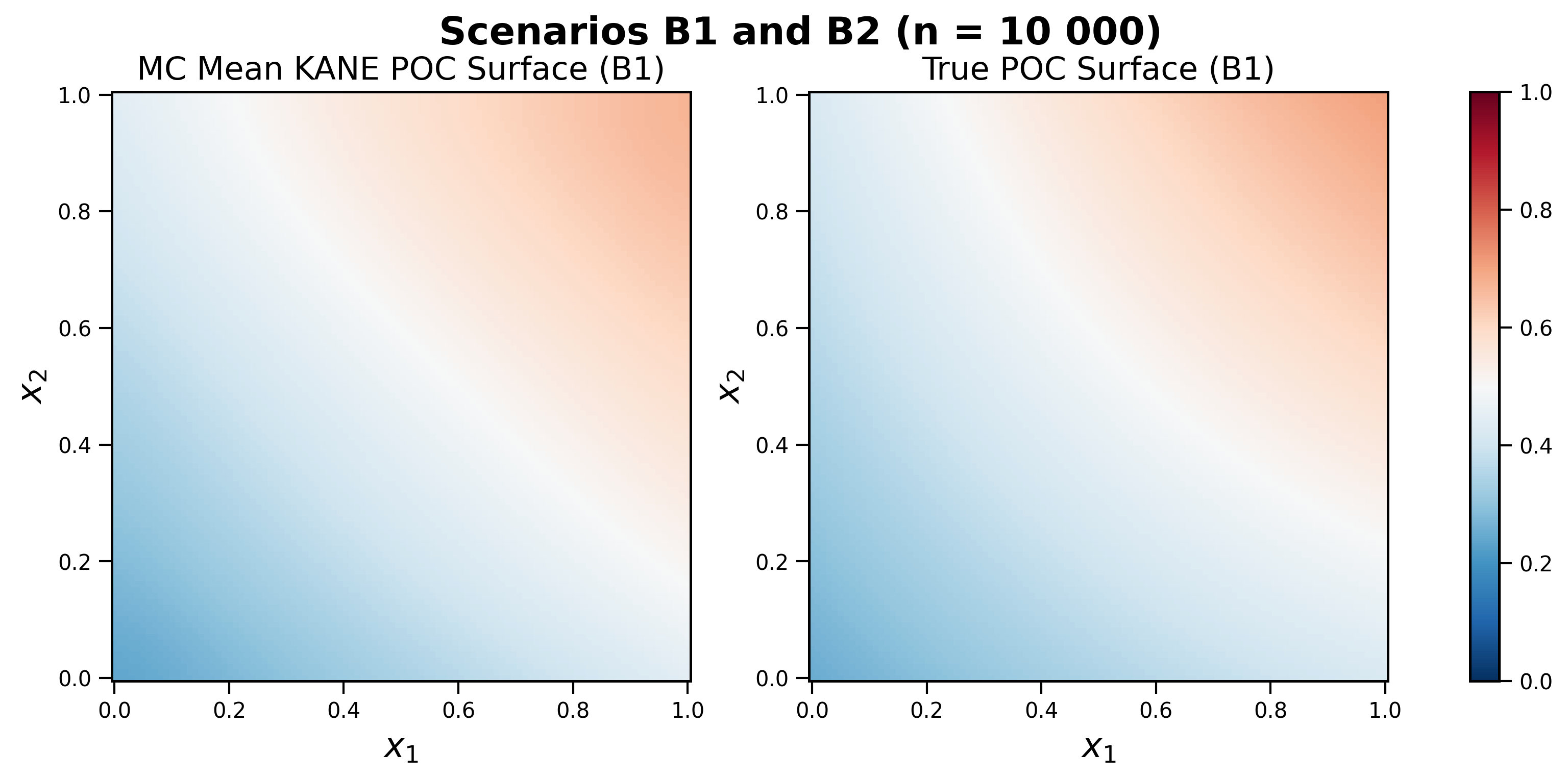}\\
    \includegraphics[width=0.95\linewidth]{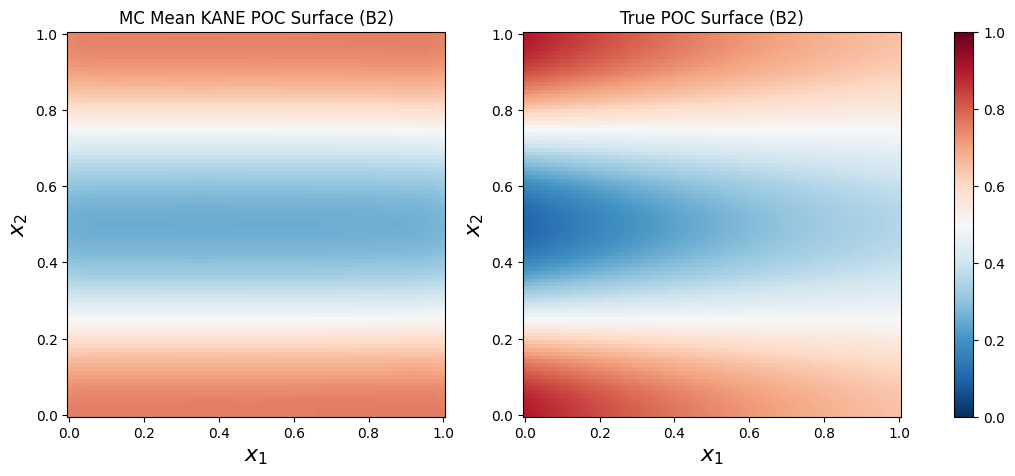}\\
    \includegraphics[width=0.95\linewidth]{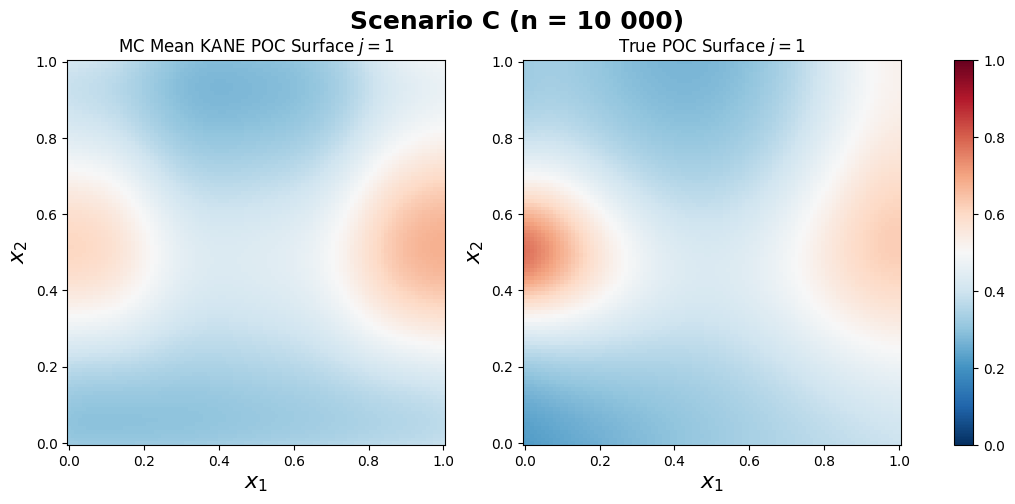}\\
    \includegraphics[width=0.95\linewidth]{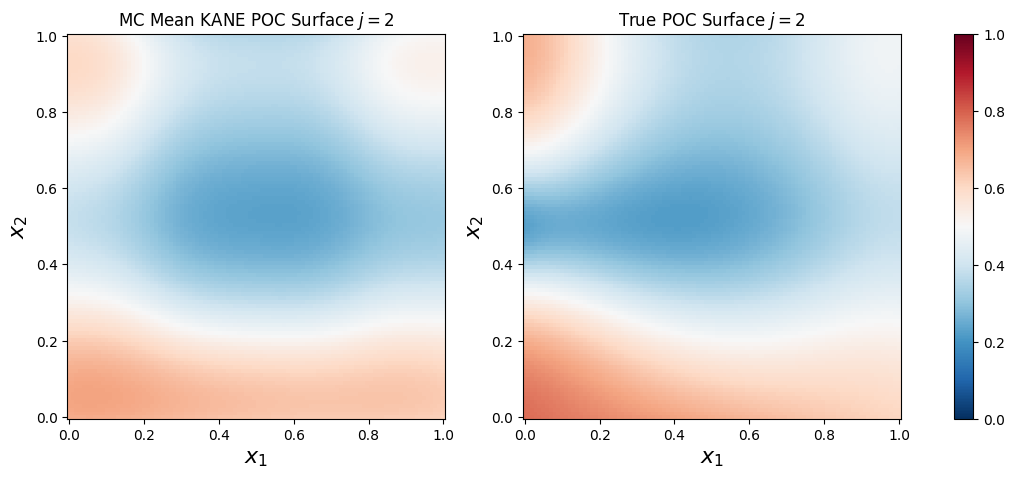}\\
    \includegraphics[width=0.95\linewidth]{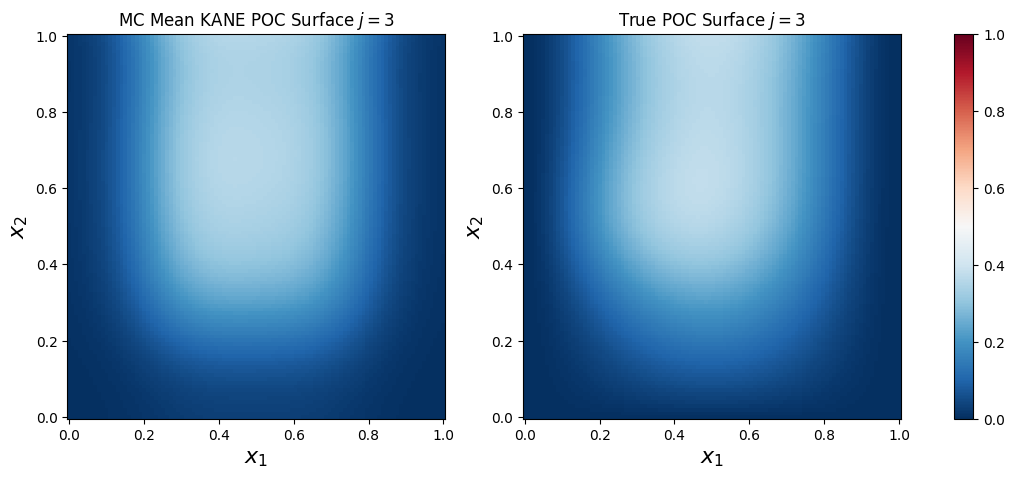}  
    \caption{Monte Carlo means for Scenarios A, B, and C.}
    \label{fig:ABCMC}
\end{figure}
\end{minipage}

\section{Empirical Examples}\label{applications}

\subsection{Earthquake--Tsunami Data}

Coastal regions are highly vulnerable to tsunamis, with their impacts often exacerbated by significant earthquakes that act as primary triggers. This section applies the proposed method to quantify the POC for tsunami occurrence ($I_{u,\mathbf{x}}$) triggered by extreme earthquakes ($Y_{\mathbf{x}}>u$).

Data were gathered from the NCEI/WDS Global Significant Earthquake Database, provided by the NOAA National Centers for Environmental Information, and consist of a point process data on earthquake locations (latitude and longitude) from 2150 B.C., including whether a tsunami followed or not. Figure \ref{fig:EDA_4.1} shows the spatial distribution of significant earthquakes and associated tsunamis.

\begin{figure}[H]
    \centering
    \includegraphics[width=0.8\linewidth]{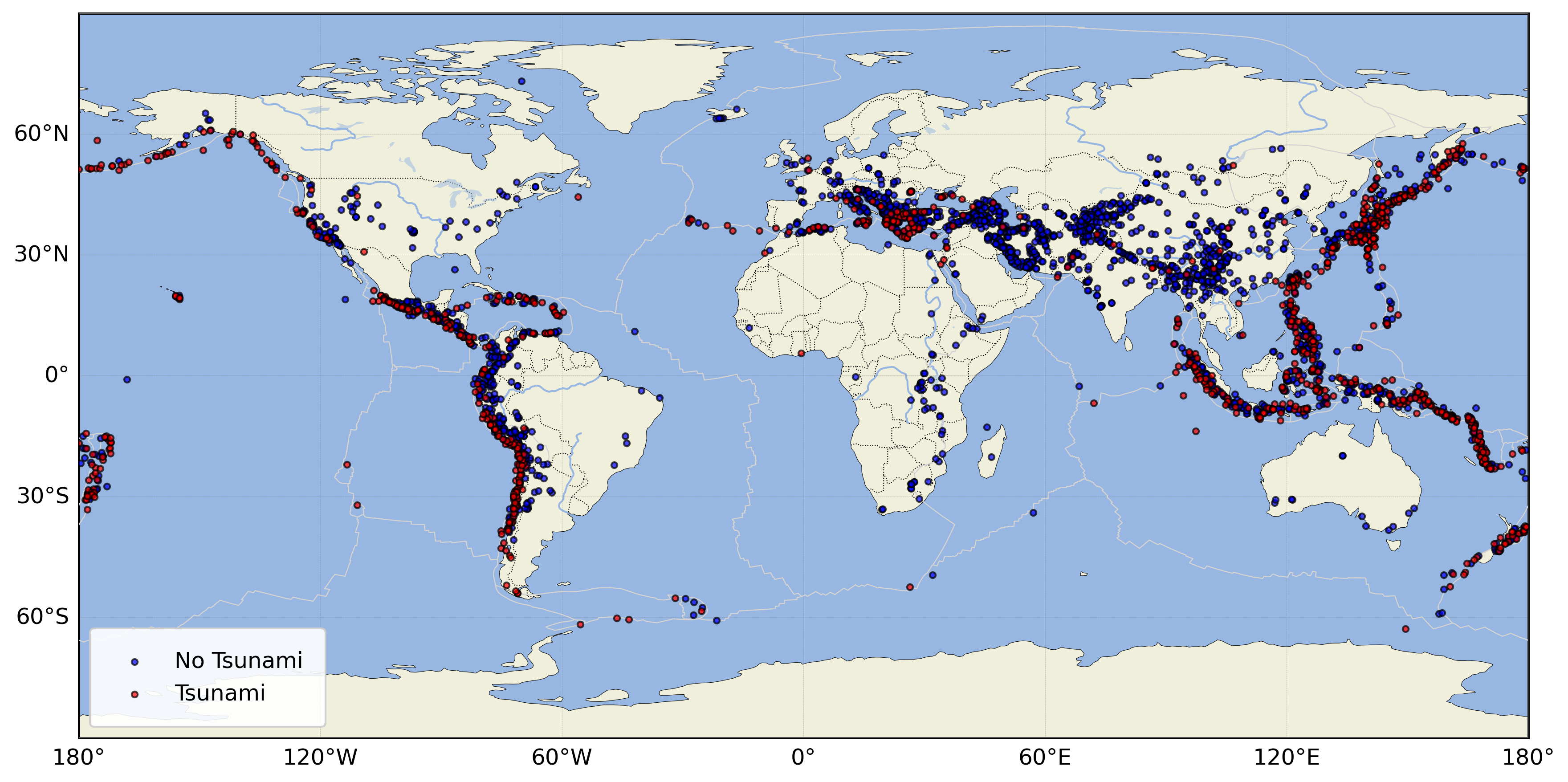}
    \caption{Point pattern of earthquakes (red) and associated tsunami occurrences (blue).}
    \label{fig:EDA_4.1}
\end{figure}

The trigger events are thresholded at $u = F_{Y_{\mathbf{x}}}^{-1}(0.95)$, resulting in 167 observations of earthquakes with magnitudes greater than 7.9 on the Richter scale, along with the corresponding indicators of tsunami occurrence. The fitted POC surfaces were modeled according to Section~\ref{KANE Bernoulli}, with a three-layer KANE framework following the setup from Section~\ref{simulation}, and considering the features
$\mathbf{x}=(\text{latitude, longitude, depth})^{T}$.

As can be seen from Fig.~\ref{fig:Application1}, the fitted surface tends to be particularly pronounced, for example in the coast of Chile, California, and the Pacific ring in comparison with the Caribbean; the same figure also illustrates that POC tends to be low at greater depths (e.g., 85th percentile), hence suggesting that earthquakes which are deeper would be less likely to lead to a cascading tsunami event. All in all, the usual extrapolation disclaimer applies. Interpretations should account for the fact that the quality of the fits is more reliable in regions with greater data availability, while their accuracy should be questioned in regions with limited data. Finally, the QQ-boxplot of the randomized residuals, shown in Fig.~\ref{fig:qqboxplots_1}, indicates an overall good fit; further trajectories in the supplementary materials indicate a similar finding. 

\begin{figure}
  \centering \footnotesize

  \hspace{0.4cm} \textbf{Depth (Percentile 15)} \hspace{5cm} \textbf{Depth (Percentile 40)} 
  \includegraphics[trim={0 3.1cm 0 0}, clip, width=0.45\linewidth]{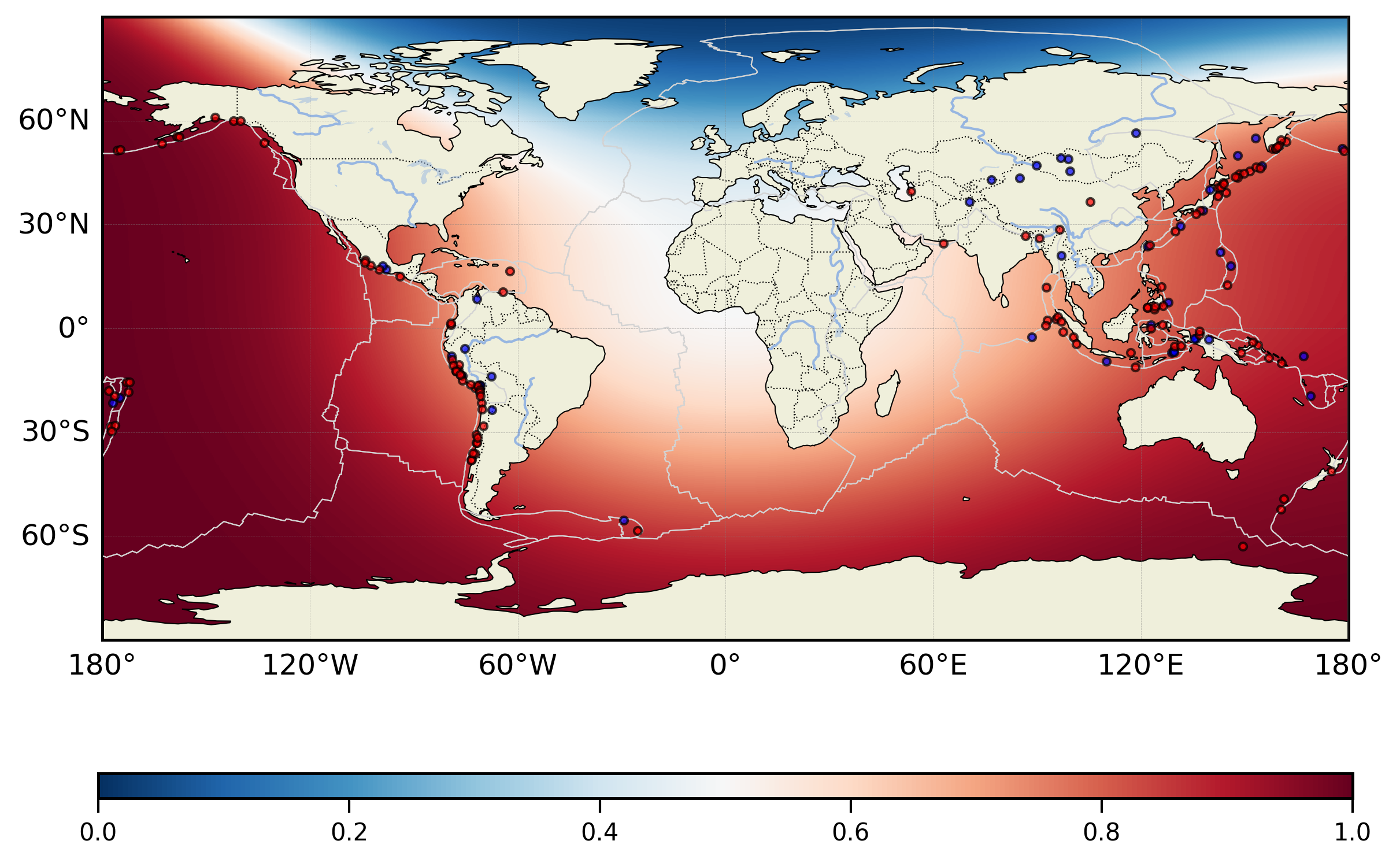}
  \includegraphics[trim={0 3.1cm 0 0}, clip, width=0.45\linewidth]{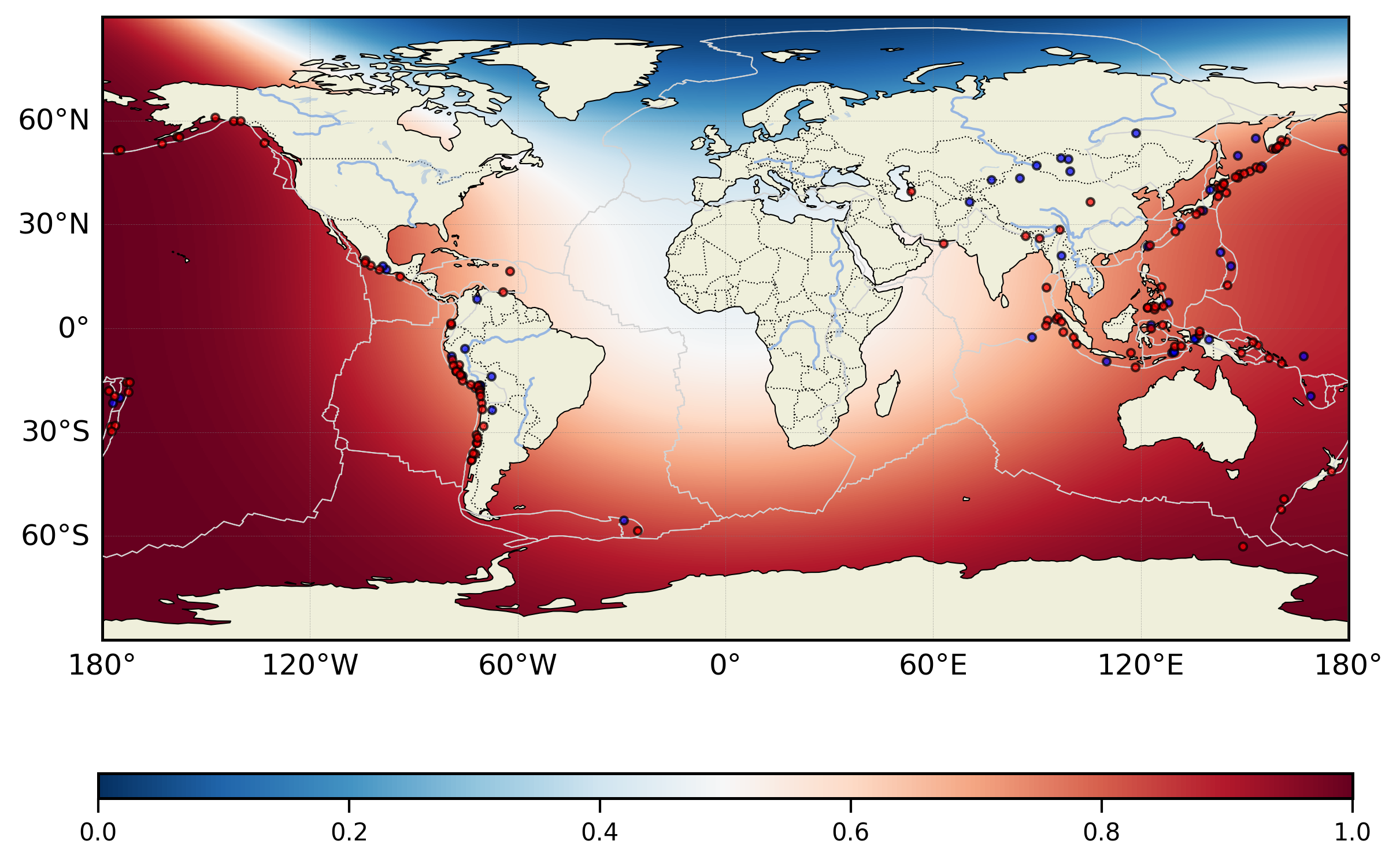} 

  \hspace{0.4cm} \textbf{Depth (Percentile 60)} \hspace{5cm} \textbf{Depth (Percentile 85)} 
  \includegraphics[trim={0 3cm 0 0}, clip, width=0.45\linewidth]{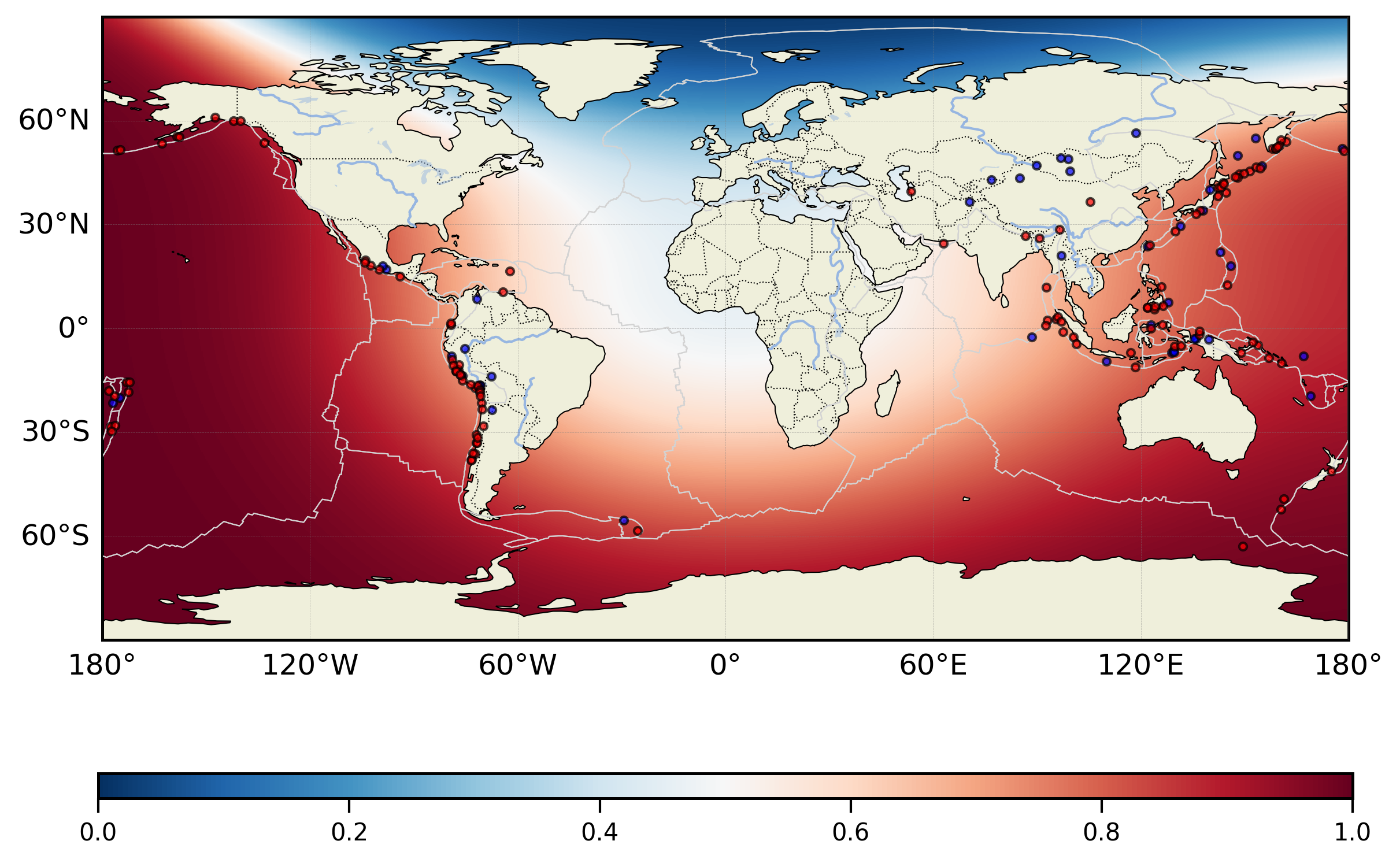}
  \includegraphics[trim={0 3cm 0 0}, clip, width=0.45\linewidth]{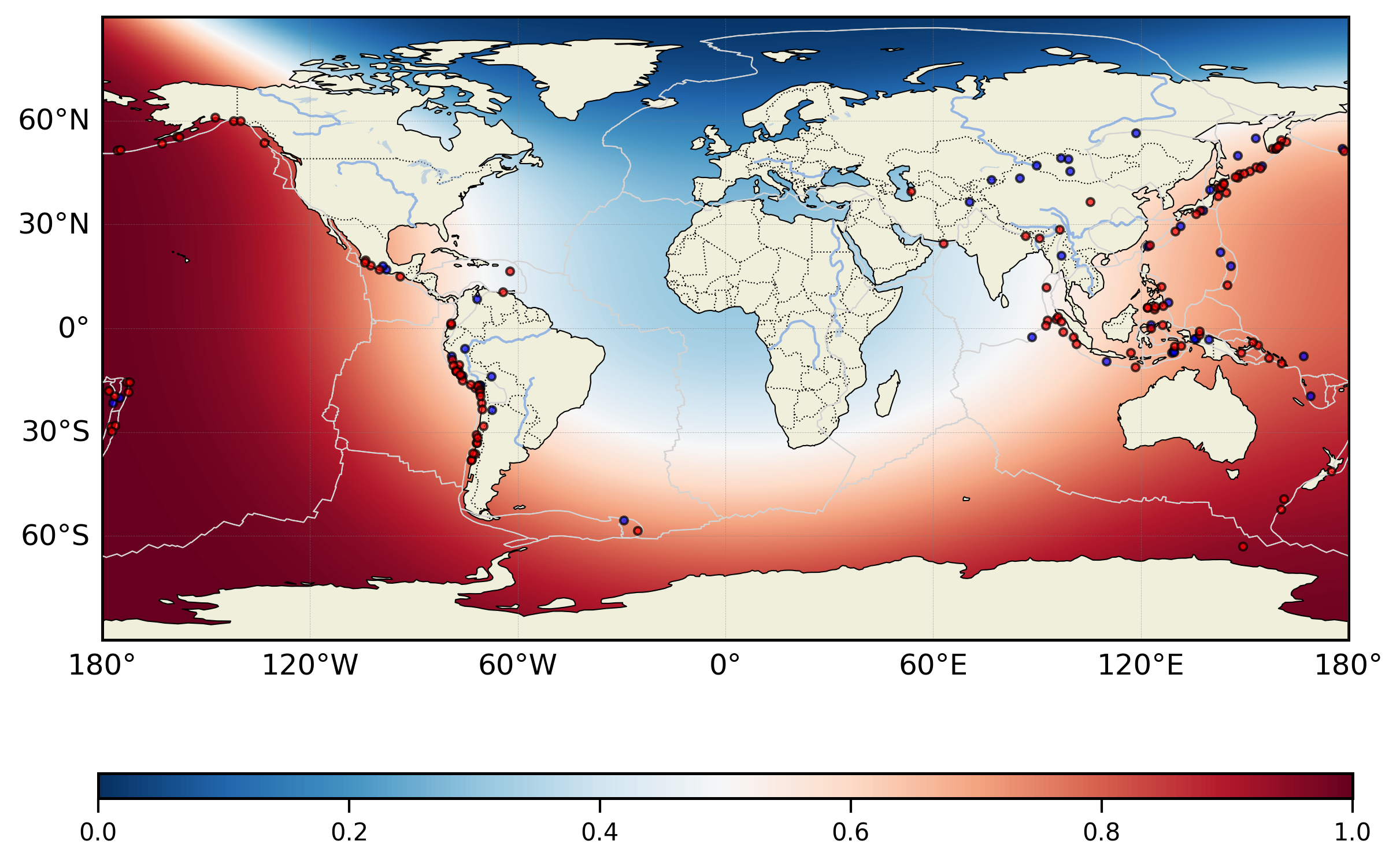} 

  \vspace{-0.1cm}  
  \hspace*{-0.4cm} \includegraphics[trim={0 0 0 10cm}, clip, width=1.02\linewidth]{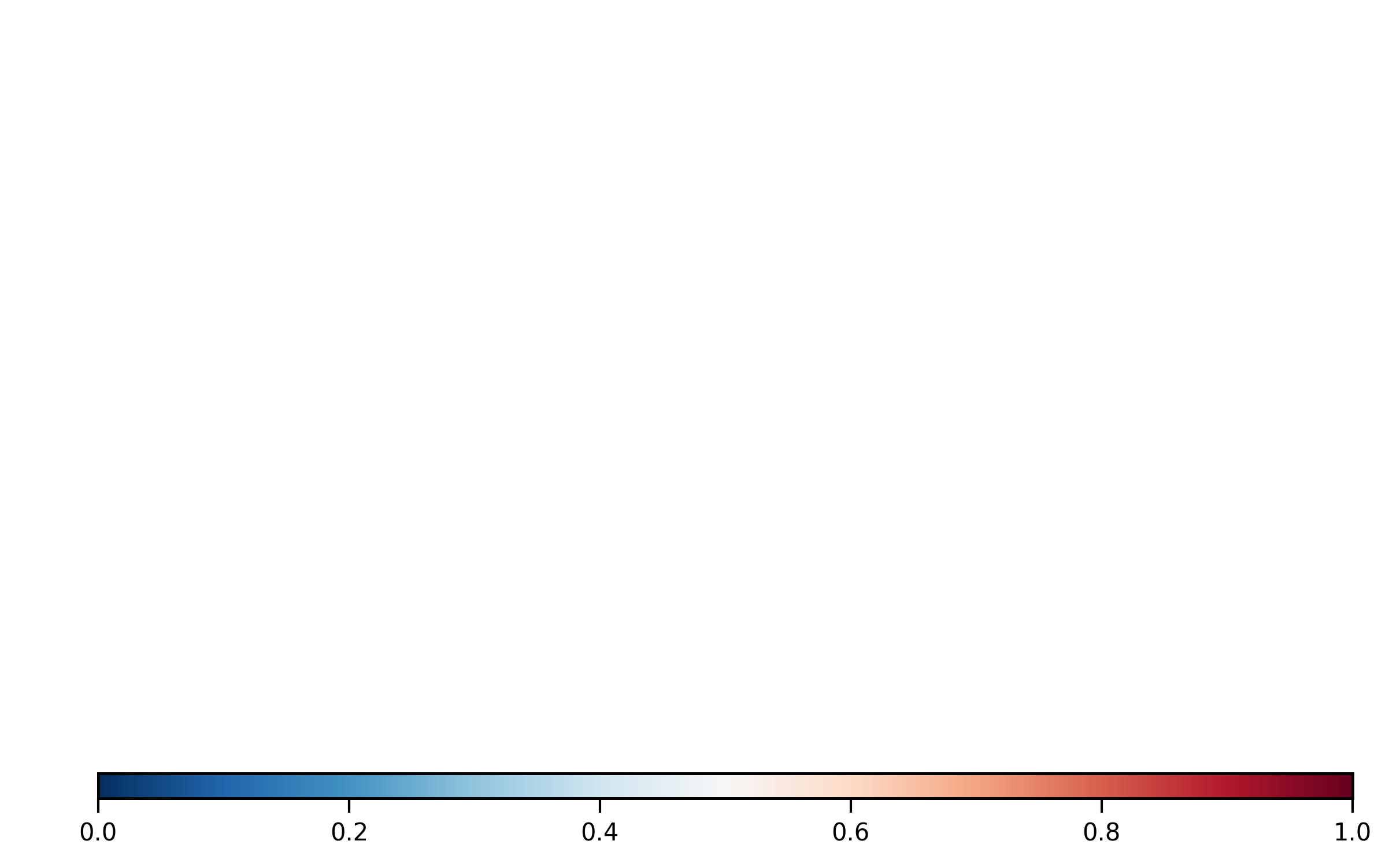} 
  \caption{KANE POC surface estimate for Earthquake--Tsunami data, considering percentile 15, 40, 60, and 85 of depth.}
  \label{fig:Application1}
\end{figure}

\subsection{Tropical Cyclone--Sea Surface Temperature Data}

Rising sea surface temperature (SST) anomalies significantly impact coastal and oceanic regions, particularly by influencing tropical cyclone intensity—categorized as Tropical Depression, Tropical Storm, and Hurricane in the North Atlantic and Northeast Pacific. Our second illustration will shed light on the ordered multicategory POC surface for tropical cyclone types ($\mathbf{I}_{u,\mathbf{x}}$), conditional on extreme SST events ($Y_{\mathbf{x}}>u$).

Data were obtained from two complementary sources covering 1981--2016: The ESA Climate Change Initiative, which provides global daily-mean SST observations on a latitude-longitude grid, and the NOAA Atlantic Hurricane dataset, which contains six-hourly storm tracks. SST values exceeding $u = F_{Y_{\mathbf{x}}}^{-1}(0.95)=302.74\,\mathrm{K}$ define extreme trigger events, yielding 456 observations. 

\begin{figure}[H]
    \centering
    \begin{tabular}{cc}
        \textbf{Tropical Depression} & \textbf{Tropical Storm} \\[2pt]
        \includegraphics[trim={3.5cm 2.5cm 4cm 0.75cm}, clip, width=0.425\textwidth]{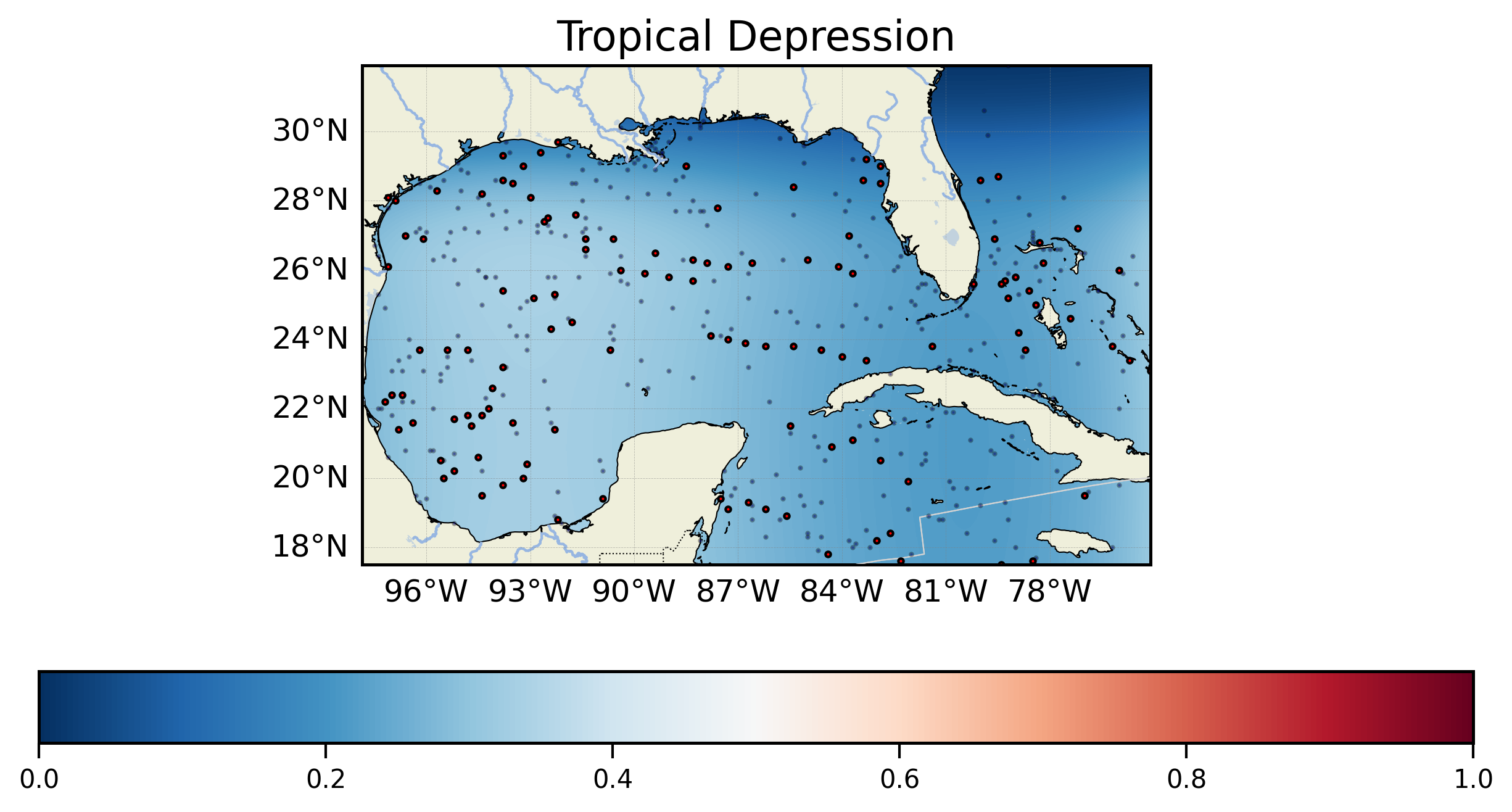} &\hspace{-0.5cm}
        \includegraphics[trim={3.5cm 2.5cm 4cm 0.75cm}, clip, width=0.425\textwidth]{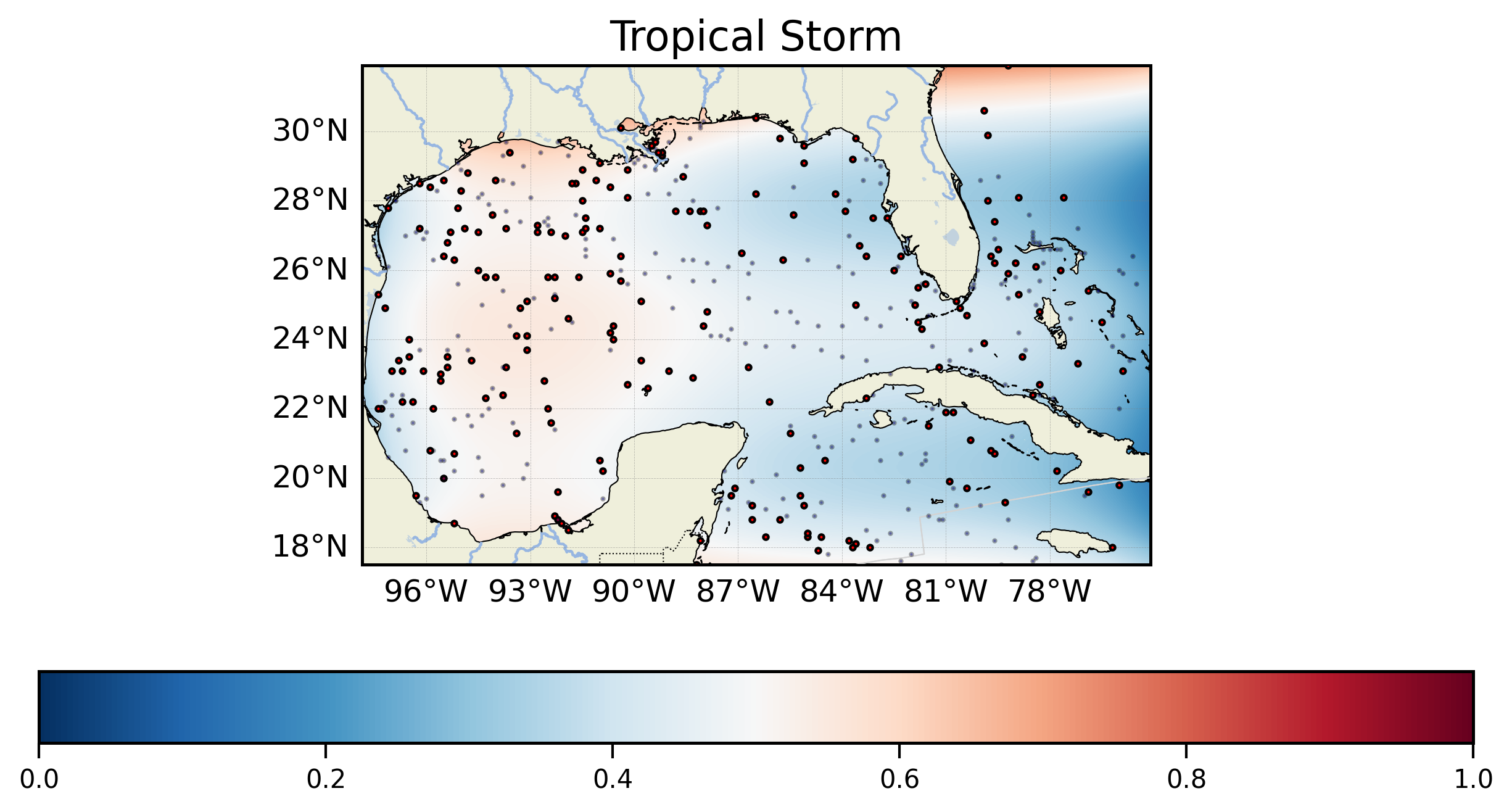} \\[5pt]
    \end{tabular}
    
    \textbf{Hurricane} \\[2pt]
    \includegraphics[trim={3.5cm 2.5cm 4cm 0.75cm}, clip, width=0.425\textwidth]{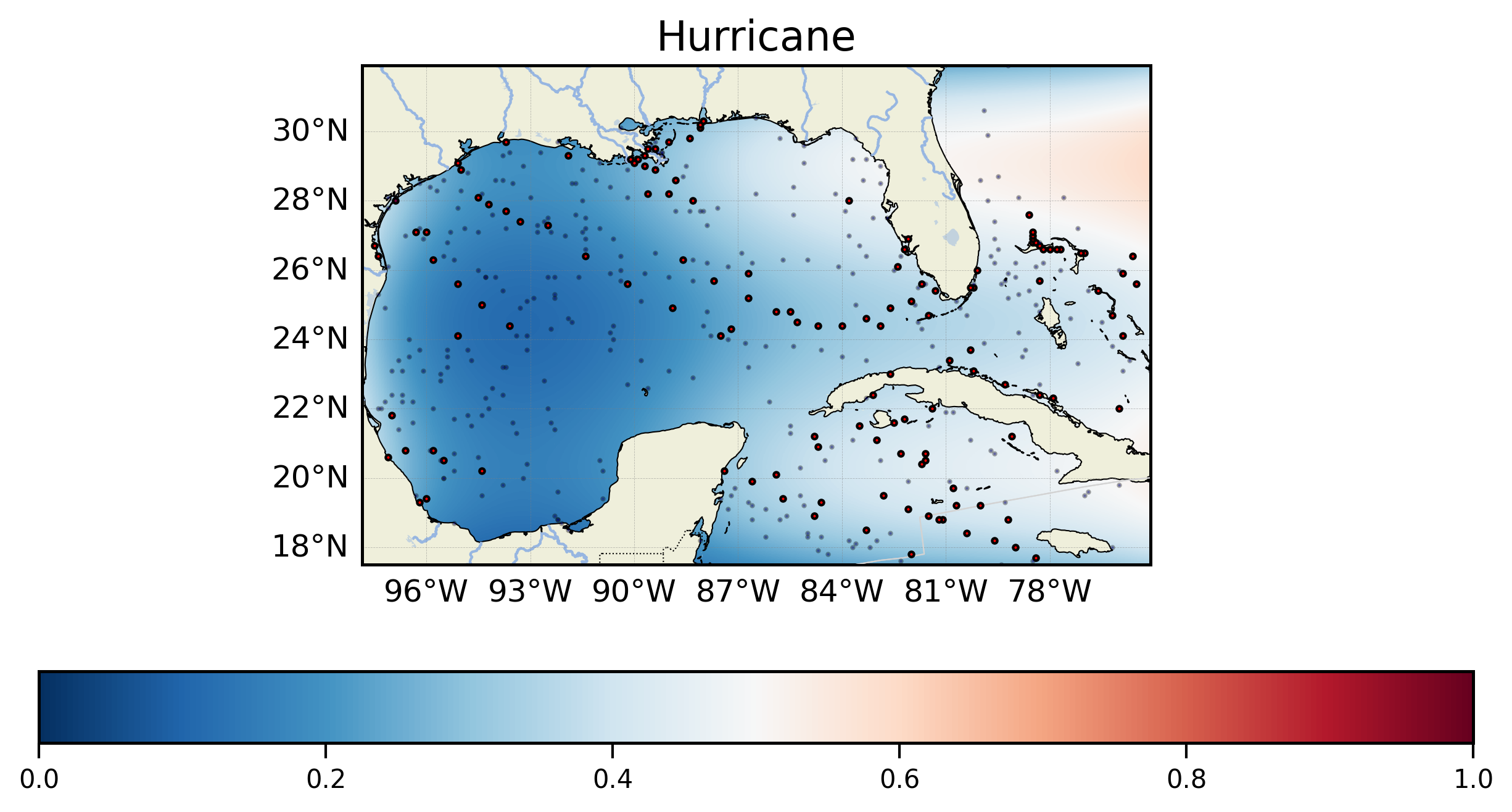} \\[5pt]

    \hspace*{-0.4cm} \includegraphics[trim={0 0 0 11cm}, clip, width=1.0\textwidth]{figures/BarColor.png} \\

    \caption{KANE POC surface estimate for Tropical Cyclone--SST data.}
    \label{Application2}
\end{figure}

Keeping in mind the ordinal nature of the data, the fitted POC surfaces were modeled according to the Frank--Hall variant of our model (Section~\ref{extensions}), with a three-layer KANE framework, and considering the features $\mathbf{x} = (\text{latitude, longitude})^{T}$. Among other insights, the fitted surfaces suggest that high sea surface temperatures tend to be associated with more extreme tropical storms in parts of the Gulf of Mexico. Aside from this, the analysis also reveals that the probability of cascade remains relatively small for tropical depression and hurricanes. Finally, the QQ-boxplot of the randomized residuals is shown in Fig.~\ref{fig:qqboxplots_1} and suggests an overall good fit of the model; further trajectories, presented in the supplementary materials, provide comparable evidence.

\begin{figure}
    \centering
    \begin{tabular}{cc}
        \textbf{Earthquake--Tsunami} & \textbf{Tropical Cyclone--SST} \\[2pt]
        \includegraphics[trim={0 1.5cm 0 0}, clip,,width=0.475\linewidth]{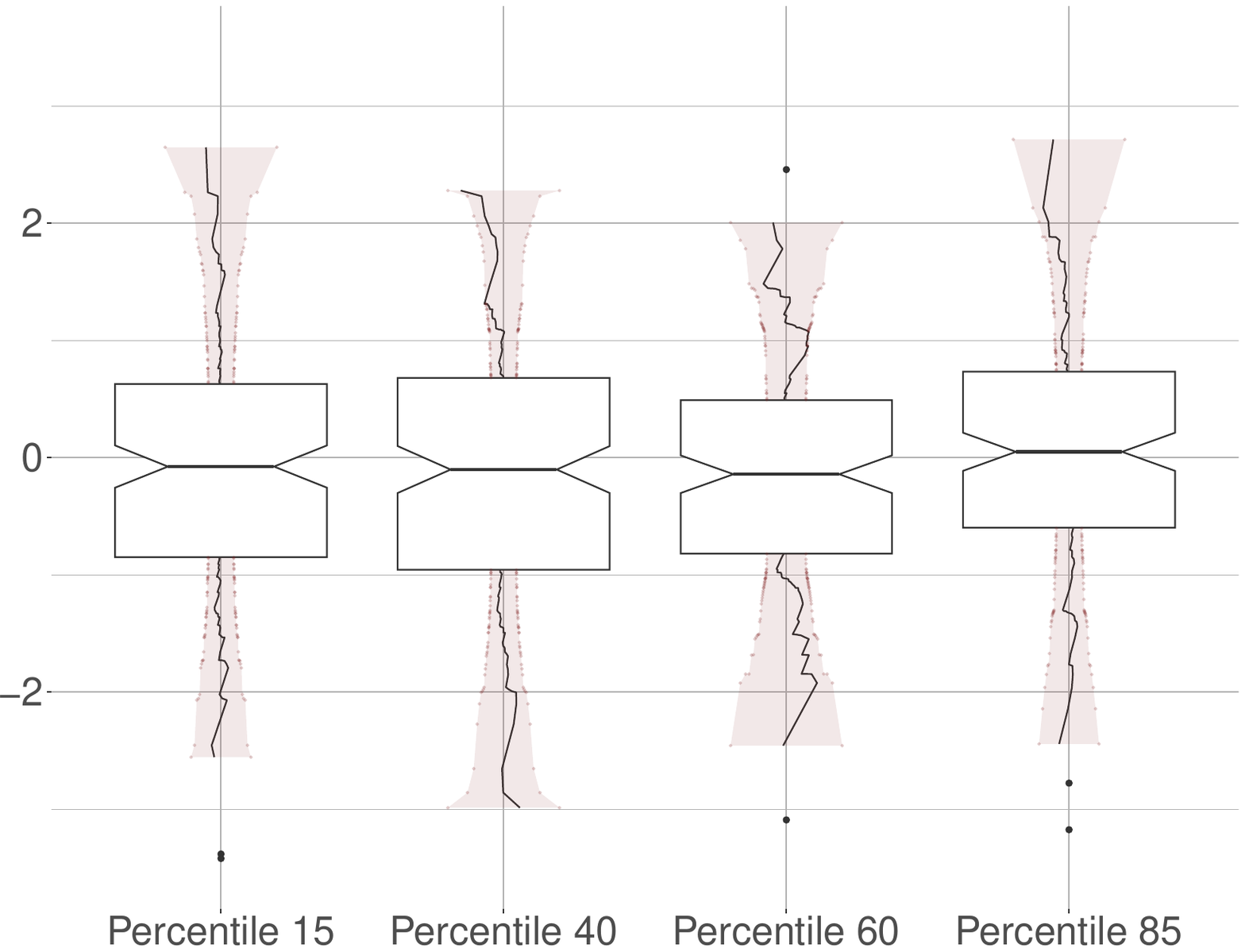} &
        \includegraphics[trim={0 1.5cm 0 0}, clip,,width=0.475\linewidth]{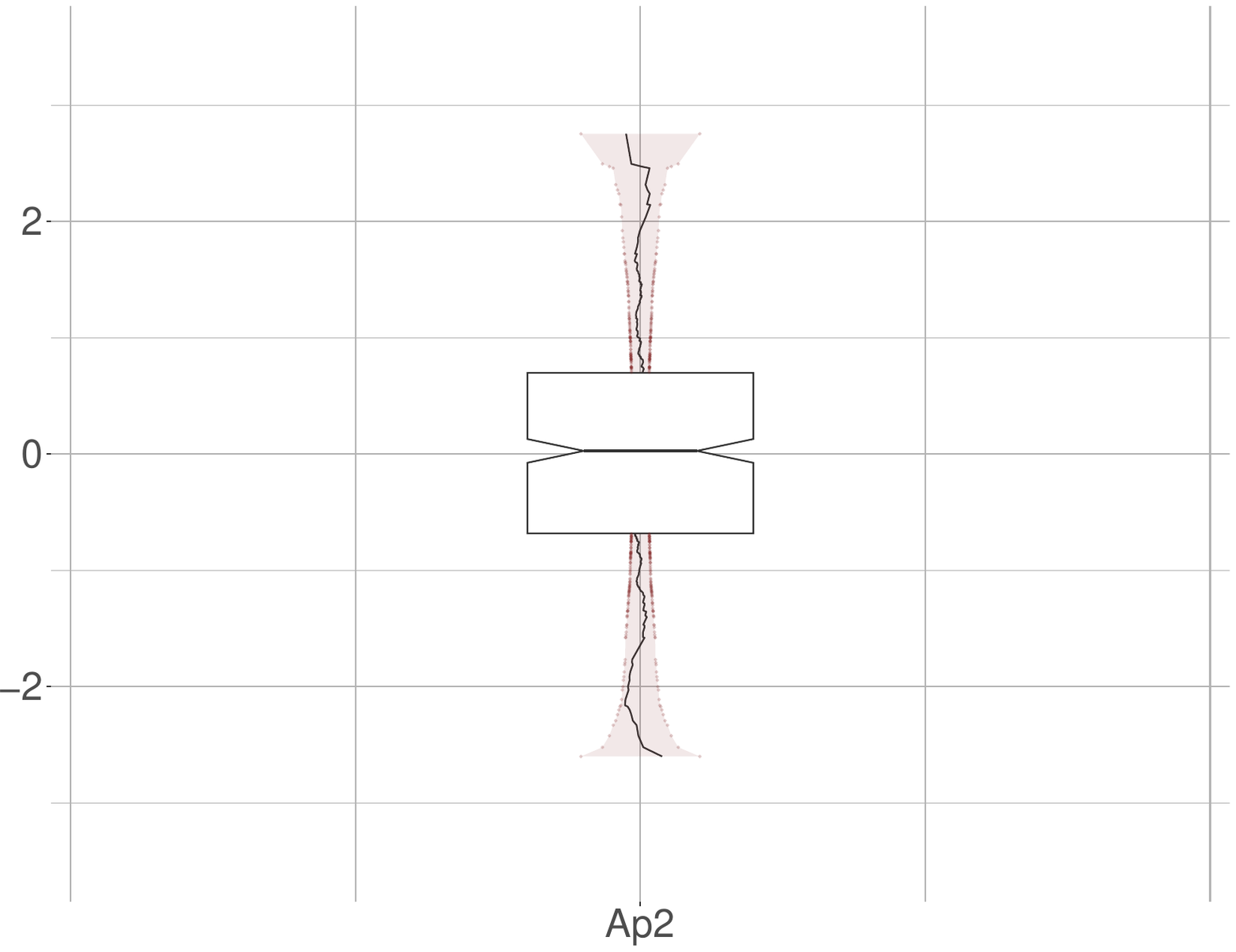}
    \end{tabular}

    \caption{QQ-boxplots of Dunn--Smyth residuals for empirical illustrations.}
    \label{fig:qqboxplots_1}
\end{figure}

\section{Implications and Future Directions}\label{discussion}
This paper develops a novel statistical framework to address the growing concern of cascading extreme events—such as earthquakes triggering tsunamis or heatwaves sparking wildfires, which in turn lead to further losses. 
By integrating EVT and AI, the proposed 
methods aim to enable extrapolation beyond observed data to capture and model complex chains of extreme 
events. Although straightforward, the additional `Natural Enforcement' of KAN discussed in this work may be valuable for other applications. For instance, it could be relevant when modeling covariate-adjusted versions of other extreme value parameters, such as $\bar{\chi}$ \citep[][Chapter~8]{coles2001}, which is constrained to the range $(-1, 1]$. This constraint could similarly be enforced naturally within our $g$-layer.

While KAN models have been claimed to be explainable \citep{liu2024}, one might consider the possibility of fitting POC surfaces using a more conventional explainable statistical model. If explainability is a priority, the POC surfaces introduced above can be formulated as an additive generalized linear models along the same lines as in \cite{lee2024}. 

We close the paper with some final comments on future research. In practice, either the follow-up or the trigger event may themselves be functional in the sense of FDA (Functional Data Analysis) \citep{horvath2012, kokoszka2017}. For example, we may observe $I_{u, \mathbf{x}} = 1$ for $\mathbf{x}$ over a continuum rather than over a point $\mathbf{x} \in [0, 1]^d$. As a concrete instance of this, in the earthquake data application, a follow-up event could represent the full region $S \subset [0, 1]^2$ affected by the tsunami, in which case $I_{u, \mathbf{x}} = 1$ for all $\mathbf{x} \in S$. While our theory also accommodates this framework, further investigation is needed to incorporate such functional events into the inferences in a fully FDA-aligned fashion. Secondly, while there are fundamental differences between our approach and Hawkes processes, there may be interest in combining our framework with the cross-excitation of multiple point processes introduced by \cite{hawkes1971}.

Finally, while the proposed approach targets two-chain setups involving a trigger and follow-up extreme event, it may set the stage for developments that accommodate extremal chains of random length ($k$), and also account for the role of the chain's pathways. TDMs (Tail Dependence Matrices) \citep{embrechts2016} may be used to govern the pathway of the extremal cascade; in this setting the order of chained events could be dictated by a random permutation of $\{1, \dots, k\}$, based on transition probabilities for the TDMs between extremal events. This would allow for modelling and learning from the data the length or size of the extremal cascade ($k$) and the distribution of the pathway taken by it (say, $4 \to 1 \to 2 \to 3$ or $3 \to 2 \to 1 \to 4$ as two examples of realizations of the chain of extremes).

We leave these open problems to future work.

\section*{Appendix}
{\subsection*{Appendix~A: The Kolmogorov Superposition Operator and its Approximation Theory}}
{As shown below small perturbations in the inner and outer 
  functions induce only small perturbations in the resulting function. This stability property may be of independent interest, and it holds under assumptions nearly as mild as those of the Kolmogorov superposition theorem itself. Hence, we refer to the operator introduced below as the Kolmogorov superposition operator.}

Throughout, $\|f\|_{\infty} \equiv \|f\|_{\infty}^A := \sup_{x \in A} |f(x)|$, and $C(A)$ and $C_{\text{Lip}}(A)$ denote the spaces of continuous and Lipschitz continuous functions on $A \subseteq \mathbb{R}$, respectively. To ease notation, let $I = \{1,\ldots,2d+1\}$ and $J = \{1,\ldots,d\}$.  Finally, we equip $C([0,1])^{(2d+1)d} \times C_{\text{Lip}}(\mathbb{R})^{2d+1}$ with the max norm 
\begin{equation*}
  \|\boldsymbol\Phi\| = \|(\boldsymbol\Phi^{(1)}, \boldsymbol\Phi^{(2)})\| = \max(m_1, m_2), \qquad 
\end{equation*}
where $$m_1 = \max_{(i, j) \in I \times J}\|\Phi^{(1)}_{i, j}\|_{\infty},
\qquad
m_2 = \max_{i \in I} \|\Phi_i^{(2)}\|_{\infty}
.$$
For notational simplicity, here and below we omit the domain in the supremum norm, though it should be understood that, for instance, $\|\Phi^{(1)}_{i, j}\|_{\infty} = \|\Phi^{(1)}_{i, j}\|_{\infty}^{[0,1]}$, $\|\Phi_i^{(2)}\|_{\infty} = \|\Phi_i^{(2)}\|_{\infty}^{\mathbb{R}}$, and similarly for other terms. The following result holds.

{\begin{theorem}[Continuity of Kolmogorov superposition operator]\label{suppoperator}
    Consider the operator
    $$\mathcal{K}: C([0,1])^{(2d+1)d} \times C_{\emph{Lip}}(\mathbb{R})^{2d+1} \;\to\; C([0,1]^d),$$
    defined by
$\mathcal{K}(\boldsymbol\Phi^{(1)}, \boldsymbol\Phi^{(2)})(x_1,\dots,x_d) \;=\; \sum_{i=1}^{2d+1} \Phi^{(2)}_i \!( \sum_{j=1}^d \Phi^{(1)}_{i,j}(x_j)),$ with 
\begin{equation*}
  \begin{split}
  \boldsymbol{\Phi}^{(1)} = \big(\Phi^{(1)}_{i,j}\big)_{(i,j)\in I\times J} \in C([0,1])^{d(2d+1)}, \qquad 
  \boldsymbol{\Phi}^{(2)} = \big(\Phi^{(2)}_{i}\big)_{i\in I} \in 
    C_{\emph{Lip}}(\mathbb{R})^{2d+1},
  \end{split}
\end{equation*}
and $I = \{1,\ldots, 2d+1\}$ and $J = \{1,\ldots,d\}$. Then, $\mathcal{K}$ is a continuous operator. 
\end{theorem}}
{\begin{proof}
Let $\tilde{\boldsymbol\Phi} = (\tilde {\boldsymbol\Phi}^{(1)}, \tilde{\boldsymbol\Phi}^{(2)}) \in C([0,1])^{(2d+1)d} \times C_{\text{Lip}}(\mathbb{R})^{2d+1}$ be arbitrary, and let $L_i$ be the Lipschitz constant for the outer function $\boldsymbol{\Phi}^{(2)}_{i}$; set $L = \max_{i \in I} L_i$. By the triangle inequality and the assumption of Lipschitz continuity of the outer functions, it follows that
\begin{align*}
\|\mathcal{K}(\boldsymbol\Phi)-\mathcal{K}(\tilde{\boldsymbol\Phi})\|_\infty
&\le
\|\mathcal{K}(\boldsymbol\Phi^{(1)},\boldsymbol\Phi^{(2)})-\mathcal{K}(\tilde{\boldsymbol\Phi}^{(1)},\boldsymbol\Phi^{(2)})\|_\infty
+\|\mathcal{K}(\tilde{\boldsymbol\Phi}^{(1)},\boldsymbol\Phi^{(2)})-\mathcal{K}(\tilde{\boldsymbol\Phi}^{(1)},\tilde{\boldsymbol\Phi}^{(2)})\|_\infty\\
&=
\|\boldsymbol\Phi^{(2)}\circ\boldsymbol\Phi^{(1)}-\boldsymbol\Phi^{(2)}\circ\tilde{\boldsymbol\Phi}^{(1)}\|_\infty
+\|\boldsymbol\Phi^{(2)}\circ\tilde{\boldsymbol\Phi}^{(1)}-\tilde{\boldsymbol\Phi}^{(2)}\circ\tilde{\boldsymbol\Phi}^{(1)}\|_\infty\\
&\le
\textstyle\sum_{i=1}^{2d+1} L_i\Bigl\|\textstyle\sum_{j=1}^{d}\bigl(\Phi^{(1)}_{i,j}-\tilde{\Phi}^{(1)}_{i,j}\bigr)\Bigr\|_\infty
+\Bigl\|\textstyle\sum_{i=1}^{2d+1}\bigl(\Phi^{(2)}_{i}-\tilde{\Phi}^{(2)}_{i}\bigr)\Bigr\|_\infty\\
&\le
\textstyle\sum_{i=1}^{2d+1} L_i\textstyle\sum_{j=1}^{d}\|\Phi^{(1)}_{i,j}-\tilde{\Phi}^{(1)}_{i,j}\|_\infty
+\textstyle\sum_{i=1}^{2d+1}\|\Phi^{(2)}_{i}-\tilde{\Phi}^{(2)}_{i}\|_\infty\\
&\le L(2d+1)d\,m_1^*+(2d+1)m_2^*.
\end{align*}

where $$m_1^* = \max_{(i, j) \in I \times J}\|\Phi^{(1)}_{i, j} - \tilde \Phi^{(1)}_{i, j}\|_{\infty}
,\qquad
m_2^* = \max_{i \in I} \|\Phi_i^{(2)} - \tilde \Phi_i^{(2)}\|_{\infty}
.$$
Thus, to achieve $\|\mathcal{K}(\boldsymbol\Phi) - \mathcal{K}(\tilde {\boldsymbol\Phi})\|_{\infty}
< \varepsilon$ for any $\varepsilon > 0$, it suffices to take $\|\boldsymbol{\Phi} - \tilde {\boldsymbol\Phi}\| < \delta$, with 
  \begin{equation*}
    \delta = \frac{\varepsilon}{(2 d + 1)(Ld + 1)}. 
  \end{equation*}
\end{proof}}
{Continuity of Kolmogorov superposition operator allows
  theoretical guarantees for the inner and outer functions to be
  translated directly into guarantees for the target function of
  interest. In particular, if the inner and outer functions can be
  approximated to a prescribed accuracy (for instance, by splines),
  then the induced approximation error on the target function (e.g.,
  the POC surface) can be controlled.}

{Since our model for the POC includes an additional layer (the $g$-layer), one may wonder about the approximation-theoretic implications of Theorem~\ref{suppoperator} for this extended architecture. For instance, under the mild assumption that $g \in C_{\mathrm{Lip}}(\mathbb{R})$ with Lipschitz constant $M > 0$, the addition of the $g$ layer does not compromise the approximation guarantees. Indeed, if $g \in C_{\mathrm{Lip}}(\mathbb{R})$ and 
  \begin{equation*}
    \|\mathcal{K}(\boldsymbol\Phi) - \mathcal{K}(\tilde{\boldsymbol\Phi})\|_{\infty}
    < \varepsilon / M,    
  \end{equation*}
then it follows immediately that
  \begin{equation}\label{glip}
    \|g \circ \mathcal{K}(\boldsymbol\Phi) - g \circ \mathcal{K}(\tilde {\boldsymbol\Phi})\|_{\infty}
    \leq M \|\mathcal{K}(\boldsymbol{\Phi}) - \mathcal{K}(\tilde {\boldsymbol\Phi})\|_{\infty}
    < \varepsilon. 
  \end{equation}
  The following example illustrates the consequences of Theorem~\ref{suppoperator} for our splines-based approach.
}
\begin{example}[Splines]
    Classical results in spline approximation theory
    \citep[e.g.,][Theorem~18]{kunoth2018} imply that any $h:A \to \mathbb{R}$ in the Sobolev space $W_{\infty}^{1}(A)$ can be uniformly approximated by a spline of degree $p$, that is,
\begin{equation}\label{splinesres}
    \|h - \tilde h\|_{\infty} < \gamma.
\end{equation}
Here, $\gamma > 0$ depends on $m$ and $p$ (as well as the smoothness of $h$), and it can 
be made arbitrarily small by considering a denser mesh. It is well-known that  $W_{\infty}^{1}(A) = C_{\text{Lip}}(A)$ \citep[][Section~5.8]{evans2022}. Thus, 
combining Theorem~\ref{suppoperator} with \eqref{splinesres} shows that any continuous 
multivariate function can be uniformly approximated if each inner and 
outer function is approximated by splines. Indeed, the argument in the proof of Theorem~\ref{suppoperator}, implies that setting $\|\boldsymbol\Phi - \tilde{\boldsymbol\Phi}\| < \delta$, with $L = \max_{i \in I} L_i$ and 
  \begin{equation*}
    \delta = \frac{\varepsilon}{M (2 d + 1)(Ld + 1)}, 
  \end{equation*}
yields
$\|\mathcal{K}(\boldsymbol{\Phi}) - \mathcal{K}(\tilde{\boldsymbol\Phi})\|_{\infty}
< \varepsilon / M$,
whenever $\tilde{\boldsymbol\Phi} = (\tilde{\boldsymbol\Phi}^{(1)},\tilde{\boldsymbol\Phi}^{(2)})$ consists of spline approximations, 
since splines satisfy \eqref{splinesres}. Finally, if we assume $g \in C_{\mathrm{Lip}}(\mathbb{R})$ with Lipschitz constant $M > 0$, then \eqref{glip} holds. \strut \hfill $\square$
\end{example}

\subsection*{Appendix B: Deep KANE}
A deep KANE model for the POC surface can be readily obtained by extending the line of attack of \cite{liu2024}. That is, the deep KANE formulation is based on the following $\ell$-layer specification, 
$$\alpha_{I}(\mathbf{x}) = {g}\left(({\boldsymbol\Phi^{(\ell - 1)}} \circ \cdots \circ {\boldsymbol\Phi^{(1)}})(\mathbf{x})\right),$$
where for $l \in \{1, \dots, \ell - 1\}$,
\begin{equation*}
  \boldsymbol\Phi^{(l)}=\begin{pmatrix}
    \Phi_{1,1}^{(l)} &  \cdots & \Phi_{1,n_l}^{(l)} \\
    \vdots & \ddots & \vdots\\
    \Phi_{n_{l + 1},1}^{(l)} & \cdots & \Phi_{n_{l + 1},n_l}^{(l)}\\
  \end{pmatrix}. 
\end{equation*}
Here, $n_l$ is the number of nodes in the $l$th layer, for $l = 1, \dots, \ell$. The three-layer KANE presented above is a particular case with $n_1 = d$, $n_2 = 2d + 1$, and $n_3 = 1$. For modeling in the deep framework we again use the specification in \eqref{hi}. The parameter of interest in this case is given by the following collection of matrices
\begin{equation*}
  \boldsymbol\beta^{(l)}_k = 
\begin{pmatrix}
\beta_{n_{l},1, k}^{(l)} &  \cdots & \beta_{n_{l},d, k}^{(l)} \\
\vdots & \ddots & \vdots \\
\beta_{n_{l + 1},1, k}^{(l)} & \cdots & \beta_{n_{l + 1},d, k}^{(l)}
\end{pmatrix},
\end{equation*}
where $k = 1, \dots, K$ and $l = 1, \dots, \ell - 1$. 

{\footnotesize
  \noindent \textbf{Acknowledgements}~We thank the Editor, Associate Editor, and two reviewers for their valuable feedback.} We thank Johnny Myung Won Lee and participants of GAME 2025 (Generative AI Modeling for Extreme Events) for insightful discussions. Finally, we thank members of GAIL (UoE Generative AI Lab), ECFI (Edinburgh Centre for Financial Innovations), and of \text{GLE$^2$N} (The Glasgow--Edinburgh Extremes Network) for  constructive comments and feedback. {} MdC is partially supported by the Royal Society of Edinburgh, Leverhulme Trust, and by Aberdeen Investments via the UoE Centre for Investing Innovation. Finally, funding from Funda\c{c}\~{a}o 
para a Ci\^{e}ncia e a Tecnologia (under grants UID/4106/2025 and UID/PRR/4106/2025) is gratefully acknowledged.
 
  \noindent \textbf{Conflicts of interest}~The authors declare that they have no conflict of interest. \vspace{0.2cm}
  
  \noindent \textbf{Data availability statement}~The datasets are available from the authors and will be included in the \textsf{R} package \texttt{DATAstudio} upon the acceptance of this paper (\url{https://cran.r-project.org/web/packages/DATAstudio/}). \vspace{0.2cm}

\noindent \textbf{Software}~Codes are publicly available from Github (\url{https://edin.ac/4gJxP2U}). \vspace{0.2cm}
}

\end{document}